%% file: StakeDag.tex
\documentclass{article}

\usepackage{arxiv}
\usepackage[utf8]{inputenc}

\usepackage{longtable}
\usepackage{authblk}
\usepackage[caption = false]{subfig}
\usepackage{graphicx}
\usepackage{amssymb}
\usepackage{amsmath}
\usepackage{amsthm}
\usepackage[noend]{algpseudocode}
\usepackage{algorithm}
\usepackage{titletoc}
\usepackage{titlesec}
\usepackage[utf8]{inputenc}
\usepackage{enumerate}

\makeatletter
\renewcommand{\Function}[2]{%
	\csname ALG@cmd@\ALG@L @Function\endcsname{#1}{#2}%
	\def\jayden@currentfunction{#1}%
}
\newcommand{\funclabel}[1]{%
	\@bsphack
	\protected@write\@auxout{}{%
		\string\newlabel{#1}{{\jayden@currentfunction}{\thepage}}%
	}%
	\@esphack
}
\makeatother

\newcommand{\stakedag}{{\rm StakeDag}}

\newtheorem{thm}{Theorem}[section]
\newtheorem{lem}[thm]{Lemma}
\newtheorem{prop}[thm]{Proposition}

\usepackage[english]{babel}

\newcommand{\dfnn}[2]{ \textbf{\emph{[#1]}} {#2}}

\newcommand{\eself}{\hookrightarrow^{s}}
\newcommand{\eref}{\hookrightarrow^{r}}
\newcommand{\eancestor}{\hookrightarrow^{a}} 
\newcommand{\eselfancestor}{\hookrightarrow^{sa}} 
\newcommand{\erefz}{\hookrightarrow}
\newcommand{\efork}{\pitchfork}

\newcommand{\hibefore}{\mapsto}
\newcommand{\hbefore}{\rightarrow}
\newcommand{\concur}{\parallel}

\makeatletter
\def\BState{\State\hskip-\ALG@thistlm}
\makeatother

\title{StakeDag: Stake-based Consensus for Scalable Trustless Systems}

\author{Quan Nguyen, Andre Cronje, Michael Kong, Alex Kampa, George Samman}
\affil{FANTOM}

\begin{document}
\maketitle

\begin{abstract}
Trustless systems, such as those blockchain enpowered, provide trust in the system regardless of the trust of its participants, who may be honest or malicious. Proof-of-stake (PoS) protocols and DAG-based approaches have emerged as a better alternative than the proof of work (PoW) for consensus.

This paper introduces a new model, so-called \emph{\stakedag}, which aims for PoS consensus in a DAG-based trustless system. We address a general model of trustless system in which participants are distinguished by their stake or trust: users and validators. Users are normal participants with a no assumed trust and validators are high profile participants with an established trust.

We then propose a new family of stake-based consensus protocols $\mathfrak{S}$, operating on the DAG as in the Lachesis protocol~\cite{lachesis01}.
Specifically, we propose a stake-based protocol $S_\phi$ that leverages participants' stake as  validating weights to achieve more secure distributed systems with practical Byzantine fault tolerance (pBFT) in leaderless asynchronous Directed Acyclic Graph (DAG).
We then present a general model of staking for asynchronous DAG-based distributed systems.
\end{abstract}


% Table of Contents
\newpage
\pagenumbering{arabic} 
\tableofcontents 
% \linenumbers
\newpage

\input{intro}

\input{related}

\input{stakedagmodel}

\input{staking}

\input{Sprotocol}

\input{discuss}

\section{Conclusion}\label{se:con}
In this paper, we introduce a new set of protocols, namely \stakedag, for a scalable asynchronous distributed system with practical BFT. We propose a new family of consensus protocols $\mathfrak{S}$ that uses Proof of Stake to achieve more scalable and robust consensus in a DAG.
Further, we present a model of staking used for our \stakedag\ framework. Weight models and staking choices have been described.

We then introduce a specific consensus protocol, called $S_\phi$, to address a more reliable consensus compared to predecedent DAG-based approaches.
The new consensus protocol $S_\phi$ uses the well-known concept of layering of DAG, like in our $L_\phi$\cite{onlay19}, to achieve a deterministic consensus on the S-OPERA chain. By using Proof of Stake, $S_\phi$ can improve the scalability, sustainability than previous DAG-based approaches.

We have included formal definitions and semantics for our general model of \stakedag.
Our formal proof of pBFT for our \stakedag\ protocol is given in the Appendix. Our work extends the formal foundation established in our previous paper~\cite{fantom18}, which is the first that studies concurrent common knowledge sematics~\cite{cck92} in DAG-based protocols. Formal proofs for our layering-based $S_\phi$ protocol is also presented.



%% References

\clearpage
\section{Reference}\label{se:ref}

\renewcommand\refname{\vskip -1cm}
\bibliographystyle{unsrt}
\bibliography{LCA}

\input{appendix}

\end{document}

%% file: intro.tex
\section{Introduction}\label{ch:intro}

Trustless systems provide trust in the system regardless of the trust of its participants. In trustless systems, the participants' trusts are not assumed, but rather they may be honest or malicious. 
After the success over cryptocurrency, blockchains have emerged as a technology platform for secure decentralized transaction ledgers. They have been applied in numerous domains including financial, logistics as well as health care sectors. 
Blockchains provide immutability and transparency of blocks and are emerging as a promising solution for building trustless systems including distributed ledgers.

The concept of \emph{Byzantine} fault tolerance (BFT)~\cite{Lamport82} guaratees the reliability of a distributed database system when one-third of the participants may be compromised. Consensus algorithms~\cite{bcbook15} ensure the integrity of transactions over the distributed network~\cite{Lamport82} and is equivalent to the proof of BFT in distributed database systems~\cite{randomized03, paxos01}. 
Deterministic, completely asynchronous system does not guarantee Byzantine consensus with unbounded delays~\cite{flp}, but it is completely feasible for nondeterministic system. 
All nodes in practical Byzantine fault tolerance (pBFT) can reach a consensus for a block in the presence of a Byzantine node \cite{Castro99}. Consensus in pBFT is reached once a created block is shared with other participants and the share information is further shared with others \cite{zyzzyva07, honey16}.

Bitcoin and blockchain technologies have shown a phenomenal success that has enabled numerous opportunities for business and innovation. These protocols facilitate highly trustworthy, append-only, transparent public distributed ledgers. The underlying technologies have brought a huge promise to shape up the future of financial transactions, and potentially to redefine how people and companies compute, work and collaborate. 
Despite of the great success, the current blockchain-based systems are still facing some challenges, which attracted a lot of research in consensus algorithms \cite{bitcoin08, ppcoin12, dagcoin15}. 

\emph{Proof of Work} (PoW)~\cite{bitcoin08} is the most used model since introduced in the original Nakamoto consensus protocol in Bitcoin. 
Under PoW, validators are randomly selected based on the computation power they use. This process costs electricity that prevents attackers to change transaction records. PoW protocol requires exhausive computational work from participants for block generation, and also needs quite a long time for transaction confirmation. Since then, there have been an extensive amount of work to address these limitations.

Recent technological and innovative advances have led to new consensus algorithms \cite{ppcoin12, dpos14, dagcoin15, algorand17, sompolinsky2016spectre, PHANTOM08} to improve the consensus confirmation time and power consumption over blockchain-powered distributed ledgers.

\emph{Proof Of Stake} (PoS)~\cite{ppcoin12,dpos14} uses participants' stakes for generating blocks. \emph{Stake} is an amount of the cryptocurrency that a participant possesses and can prove it. Under PoS, the network achieves distributed consensus on a blockchain by randomly selecting the creator of the next block with a probability based on their stake. In PoS, validators vote on the authentic transactions based on their stake, an amount of tokens that they deposit into an account, which is frozen for a certain period of time. PoS brings environmental advantage with much less power consumption than PoW. If a participant is dishonest and compromised, its stake is voided and burnt. 
The penalty of losing its stake has proved an effective prevention of attack, since the stake loss will outweight the potential gain of an attack for any attacker. Thus, PoS is safer than proof-of-work (PoW), especially due to the scarcity of stakes in PoS compared to the easy-to-get computing power required in PoW.

The concept of a \emph{DAG} (directed acyclic graph) cryptocurrency was first introduced in 2015 in DagCoin paper~\cite{dagcoin15}. DAG technology becomes a promising alternative that allows cryptocurrencies to function similarly to those that utilize blockchain technology without the need for blocks and miners.  
DAG-based approaches have recently emerged as a promising alternative than the proof of work (PoW) for consensus. 
These approaches utilize directed acyclic graphs (DAG)~\cite{dagcoin15, sompolinsky2016spectre, PHANTOM08, PARSEC18, conflux18} to facilitate consensus. 
Examples of DAG-based consensus algorithms include Tangle~\cite{tangle17}, Byteball~\cite{byteball16}, and Hashgraph~\cite{hashgraph16}.

\subsection{Motivation}
Lachesis protocol~\cite{lachesis01} presents a general model of DAG-based consensus protocols.
The Lachesis consensus protocols  create a directed acyclic graph for distributed systems. 
We introduced a Lachesis consensus protocols\cite{lachesis01,fantom18}, which are DAG-based asynchronous non-deterministic to achieve pBFT. The protocols generate each block asynchronously and uses the OPERA chain (DAG) for faster consensus by confirming how many nodes share the blocks.
Recently, we introduced an ONLAY framework~\cite{onlay19} that achieves scalable, reliable consensus in a leaderless aBFT DAG. This framework uses the concepts of graph layering and hierarchical graphs on the DAGs. Then assigned layers are used to achieve deterministic topological ordering of finalized event blocks in an asynchronous leaderless DAG-based system.

There is only a few research work studying Proof-of-Stake in DAG-based consensus protocols; one example is ~\cite{hashgraph16}.
Hence, we are interested in investigating to use participants' stakes to improve DAG-based consensus protocol. We aim to see whether such a Proof of Stake model, which associates each participant with their stake or trust can guarantee a more reliable and robust consensus in trustless systems.

\subsection{\stakedag\ Protocols}

In this paper, we propose a new family $\mathfrak{S}$ of consensus protocols, denoted by the \stakedag\ protocols, that aims for PoS based consensus in a DAG-based trustless system. 
We introduce a general model of trustless system in which participants are distinguished by their stake or trust: users and validators. Users are participants with a default low score of trust and validators are high profile with a high trust score.

A synchronous approach in BFT systems broadcasts the voting and asks each node for a vote on the validity of each block. Instead, our \stakedag\ protocol aims for an asynchronous leaderless system. \stakedag\ protocol uses the concept of distributed common knowledge together with network broadcasting to reach a consistent global view with high probability from its local view.
Each node batches client transactions into a new event block and stores in its own DAG. The new event block is then shared with other nodes through asynchronous event transmission. Each node shares its own blocks as well as the ones it received from other nodes. This asynchronous step will spread all information through the network, and thus it can increase throughput near linearly as the number of nodes participating the network.

Our general model of stake-based consensus protocols $\mathfrak{S}$, which is based on the Lachesis protocol~\cite{lachesis01}. In $\mathfrak{S}$ protocol, each event block has one self-parent reference to the top event block of the same creator, and $k$-1 other-parent references to the top blocks of other nodes.
Specifically, $\mathfrak{S}$ protocol  leverages participants' stake as  validating power to achieve practical Byzantine fault tolerance (pBFT) in leaderless asynchronous Directed Acyclic Graph (DAG).
We then present a general model of staking for asynchronous DAG-based distributed systems.

\begin{figure}[ht]
	\centering
	\includegraphics[width=0.9\linewidth]{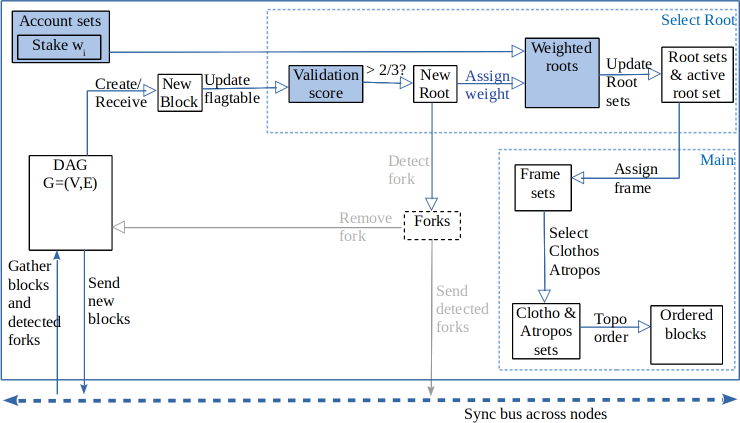}
	\caption{A General Framework of \stakedag\ Protocols}
	\label{fig:stakedagframework}
\end{figure}
Figure~\ref{fig:stakedagframework} shows a general framework of our \stakedag\ protocols. Each node contains a DAG consisting of event blocks. In each node, the information of accounts and their stakes are stored. The main steps in a PoS DAG-based consensus protocol include (a) block creation, (b) updating flagtable, (c) selecting roots and updating the root sets, (d) assigning frames, (e) selecting Clothos/Atropos, and (f) ordering the final blocks.
For a \stakedag\ protocol, the major steps are highlighted in blue. These steps are (1) updating flagtable; (2) computing validation score of a block, and (3) assigning weights to new roots. Remarkably, the Check of whether a block is a root in \stakedag\ protocol is different from that Check in Lachesis protocol: \stakedag\ protocol requires more than 2/3 of validating power (of total stake) while Lachesis requires more than 2/3 of the total number of nodes.

\begin{figure}[ht]
	\centering
	\subfloat[DAG with $k$ = 2]{\includegraphics[width=0.4\linewidth]{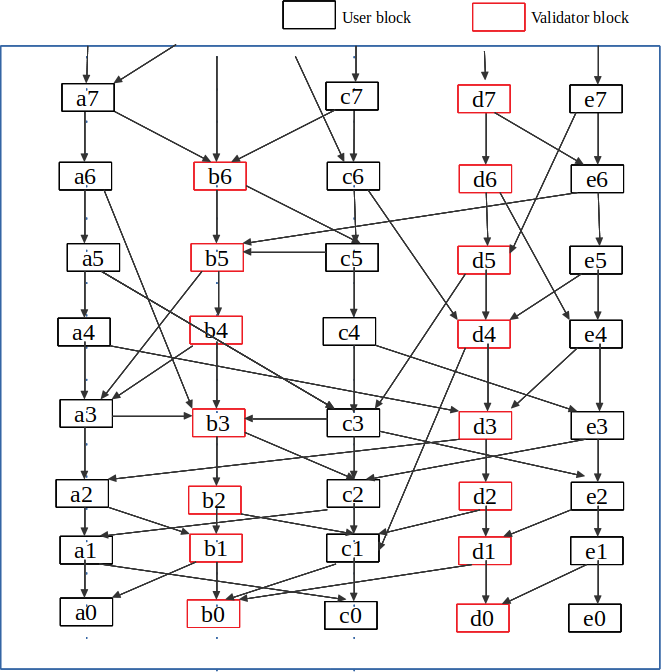}} \qquad
	\subfloat[DAG with $k$ = 3]{\includegraphics[width=0.4\linewidth]{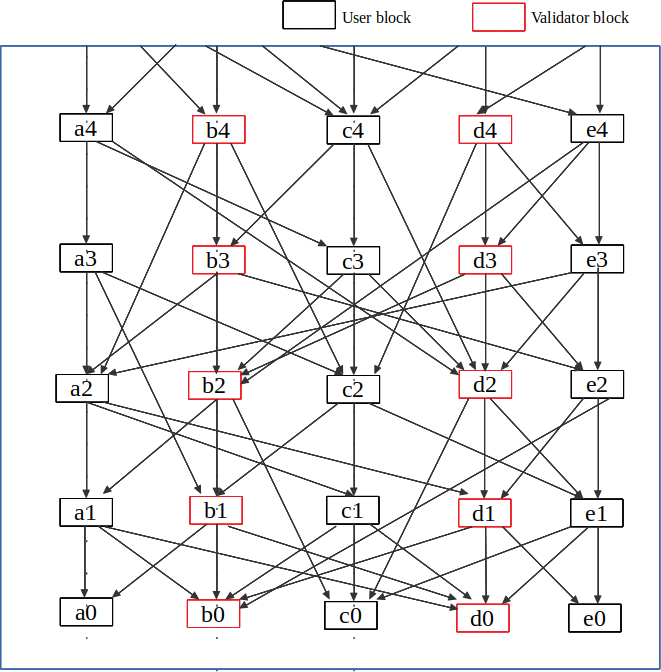}}	
	\caption{Examples of DAG with users and validators}
	\label{fig:stakedag-ex}
\end{figure}

Figure~\ref{fig:stakedag-ex} depicts an example of an S-OPERA chain, which is a weighted acyclic directed graph stored as the local view of each node.
There are five nodes in the example, three of which are normal users with validating power of 1, and the rest are validators whose validating power are set to 2. In this example, each event block has two references: self-parent and other-parent references. The event blocks created by the validators are highlighted in red.

\stakedag\ protocol leverages asynchronous event transmission for practical Byzantine fault tolerance (pBFT). 
The core idea of \stakedag\ is to create a leaderless, scalable, asynchronous DAG.
By computing asynchronous partially ordered sets with logical time ordering instead of blockchains, \stakedag\ offers a new practical alternative framework for distributed ledgers.

The main concepts of \stakedag\ are given as follows:
\begin{description}
\item[$-$ Event block] An immutable set of transactions created by a node and is then transported to other nodes.  Event block includes signature, timestamp, transaction records and referencing hashes to previous (parent) blocks.
\item [$-$ $\mathfrak{S}$ protocols] a family of \stakedag\ protocols.
\item [$-$ $S$ protocol] a specific protocol of the $\mathfrak{S}$ family, which sets the rules for event creation, communication and reaching consensus in \stakedag.
\item [$-$ Stake] This corresponds to the amount of tokens each node posesses in their deposit. This value decides the validating power a node can have.
\item [$-$ User node] A user node has a small amount stake (e.g., containing 1 token).
\item [$-$ Validator node] A validator node has large amount of stakes ($\geq$ 2 tokens).

\item [$-$ Validation score] Each event block has a validation score, which is the sum of the weights of the roots that are reachable from the block.

\item [$-$ S-OPERA chain] is the local view of the weighted Directed Acyclic Graph (DAG) held by each node. This local view is used to determine consensus.

\item [$-$ Root] An event block is called a \emph{root} if either (1) it is the first event block of a node, or (2) it can reach more than $2/3$ of the network's validating power from other roots. A \emph{root set} $R_s$ contains all the roots of a frame. A \emph{frame} $f$ is a natural number assigned to Root sets and its dependent event blocks.
\item [$-$ Root graph] Root graph contains roots as vertices and reachability between roots as edges.
\item [$-$ Clotho] A Clotho is a root at layer $i$ that is known by a root of a higher frame ($i$ + 1), and which in turns is known by another root in a higher frame ($i$ +2).
\item[$-$ Atropos] An Atropos is a Clotho that is assigned with a consensus time.
\item[$-$ Main chain] \stakedag's Main chain is a list of Atropos blocks and the the subgraphs reachable from those Atropos blocks.
\end{description}

We then introduce a specific $S_\phi$ protocol of the $\mathfrak{S}$ family. $S_\phi$ protocol uses the concept of layering, similar to its use in our ONLAY protocol~\cite{onlay19}.
$S$ protocol integrates online layering algorithms with stake-based validation to achieve practical Byzantine fault tolerance (pBFT) in leaderless DAG. The protocol achieves reliable, scalable consensus in asynchronous pBFT by using assigned layers and asynchronous partially ordered sets with logical time ordering instead of blockchains. The partial ordering produced by $S_\phi$ is flexible but consistent across the distributed system of nodes.  

We then present a formal model for the $S_\phi$ protocol. The formalization can be applied to abstract asynchronous Proof of Stake DAG-based distributed system. The formal model is built upon the model of current common knowledge (CCK)~\cite{cck92}. 

\subsection{\stakedag\ Staking Model}
In this paper, we present a staking model, which can be applied to a general model of stake-based PoS consensus protocols. In particular, our staking model can be integrated nicely with our \stakedag\ protocols.

We introduce the mechanics, mathematical reasoning and formulae for calculating various aspects of \stakedag\ system.
We then present our system design with multiple incentive mechanisms to achieve high throughput, scalability, security and decentralisation.

\subsection{Contributions}

In summary, this paper makes the following contributions:
\begin{itemize}
	\item We propose a new scalable framework, so-called \stakedag, aiming for practical more secure DAG-based trustless systems. 
	\item We present a family $\mathfrak{S}$ of PoS DAG-based consensus protocols to achieve more reliable consensus in asynchronous leaderless trustless systems.
	\item We present a staking model that can be applied to any PoS DAG-based protocols.
	\item We introduce a novel consensus protocol $S_\phi$, which uses layer algorithm and root graphs, for faster root selection. $S_\phi$ protocol uses layer assignment on the DAG to achieve quick consensus with a more reliable ordering of final event blocks.
	\item A formal model and proof of BFT of our \stakedag\ protocol are defined using CCK model~\cite{cck92}. We formalize our proofs into any generic aBFT Proof of Stake DAG system. \stakedag\ protocol achieves global consistent view via layer assignment with probability one in pBFT condition.
\end{itemize}

\subsection{Paper structure}

The rest of this paper is organised as follows.
Section~\ref{se:related} gives the related work. Section~\ref{se:model} presents our general model of Proof of Stake DAG-based consensus protocols.
Section~\ref{se:staking} 
introduces a staking model that is used for our \stakedag\ consensus protocol. 
Section~\ref{se:Sprotocol} describes a specific protocol, namely $S$ protocol, that uses layering algorithms to achieve a more reliable and scalable solution to the consensus problem in BFT systems.
Section~\ref{se:discuss} discusses about several important aspects of Proof of Stake DAG-based protocols, such as fairness and security.
Section~\ref{se:con} concludes.
Further details about our \stakedag\ framework and the $S_\phi$ protocol are given in the Appendix. It covers details about background plus Proof of Byzantine fault tolerance of the protocols.

%% file: related.tex
\section{Related work}\label{se:related}

\subsection{An overview of Blockchains}

A blockchain is a type of distributed ledger technology (DLT) to build record-keeping system in which its users possess a copy of the ledger. The system stores transactions into blocks that are linked together. The blocks and the resulting chain are immutable and therefore serving as a proof of existence of a transaction. 
In recent years, the blockchain technologies have seen a widespread interest, with applications across many sectors such as
finance, energy, public services and sharing platforms. 

For a public or permissionless blockchain, it has no central authority. Instead, consensus among users is paramount to guarantee the security and the sustainability of the system. For private, permissioned or consortium blockchains, one entity or a group of entities can control who sees, writes and modifies the data on it. 
We are interested in the decentralisation of BCT and hence only public blockchains are considered in this section. 

In order to reach a global consensus on the blockchain, users must follow the rules set by the consensus protocol of the system. 

\textbf{Proof of Work} 
Bitcoin was the first and most used BCT application. Proof of Work (PoW) is the consensus protocol introduced by the Bitcoin in 2008~\cite{bitcoin08}.  
In PoW protocol, it relies on user's computational power to solve a cryptographic puzzle that creates consensus and ensures the integrity of data stored in the chain. Nodes validate transactions in blocks (i.e. verify if sender has sufficient funds and is not double-spending) and competes with each other to solve the puzzle set by the protocol. The incentive for miners to join this mining process is two-fold: the first miner, who finds a solution, is rewarded (block reward) and gains all the transaction fees associated to the transactions.

The key component in Bitcoin protocol and its successors is the PoW puzzle solving. The miner that finds it first can issue the next block and her work is rewarded in cryptocurrency.
PoW comes together with an enormous energy demand. 

From an abstract view, there are two properties of PoW blockchain:
\begin{itemize}
\item Randomized leader election: The puzzle contest winner is elected to be the leader and hence has the right to issue the next block. The more computational power a miner has, the more likely it can be elected.
\item Incentive structure: that keeps miners behaving honestly and extending the blockchain. In Bitcoin this is achieved by miners getting a block discovery reward and transaction fees from users. In contrast, subverting the protocol would reduce the trust of the currency and would lead to price loss. Hence, a miner would end up undermining the currency that she itself collects.
\end{itemize}

Based on those two charateristics, several  alternatives to Proof-of-Work have been proposed.

\subsection{Proof of Stake}

\emph{Proof of Stake}(PoS) is an alternative to PoW for blockchains. Instead of using puzzle solving, PoS is more of a lottery system.  Each node has a certain amount of stake in a blockchain. Stake can be the amount of currency, or the age of the coin that a miner holds.
In PoS, the leader or block submitter is randomly elected with a probability proportional to the amount of stake it owns in the system. For each block, a randomly elected participant can issue the next block. The more stake a party has, the more likely it can be elected as a leader. Similarly to PoW, block issuing is rewarded with transaction fees from participants whose data are included.

There are two major types of PoS. The first type is chain-based PoS~\cite{pass2017fruitchains}, which uses chain of blocks like in PoW, but stakeholders are randomly selected based on their stake to create new blocks. This includes Peercoin~\cite{king2012ppcoin}, Blackcoin~\cite{vasin2014blackcoin}, and Iddo Bentov’s work~\cite{bentov2016}, just to name a few. The second type is BFT-based PoS that is based on BFT consensus algorithms such as pBFT~\cite{Castro99}. Proof of stake using BFT was first introduced by Tendermint~\cite{kwon2014tendermint}, and has attracted more research~\cite{algorand16}. 
Ethereum had a plan to move from a PoW to a PoS blockchain~\cite{buterin2018}. 

PoS has a clear benefit (over PoW) since
any node can join the network with even on cheap hardware such as a \$35 Raspberry Pi.
Every validator can submit blocks and the likelihood of acceptance is proportional to the \% of network weight (i.e., total amount of tokens being staked) they possess. Thus, to secure the blockchain, nodes need the actual native token of that blockchain. To acquire the native tokens, one has to purchase or earn them with staking rewards. Generally, gaining 51\% of a network’s stake is much harder than renting computation.

{\bf Security} Although PoS approach reduces the energy demand, new issues arise that were not present in PoW-based blockchains. These issues are shown as follows:
\begin{itemize}
	\item {\it Grinding attack:} Malicious nodes can play their bias in the election process to gain more rewards or to double spend their money.
	\item {\it Nothing at stake attack:} In PoS,  constructing alternative chains becomes easier. A node in PoS seemingly does not lose anything by also mining on an alternative chain, whereas it would lose CPU time if working on an alternative chain in PoW.
\end{itemize}

{\bf Delegated Proof of Stake}
To tackle the above issues in PoS, 
\emph{Delegated Proof of Stake} (DPoS) consensus protocols are introduced, such as Lisk, EOS~\cite{EOS}, Steem~\cite{Steem}, BitShares~\cite{bitshares} and Ark~\cite{Ark}. 
DPoS uses voting to reach consensus among nodes more efficiently by speeding up transactions and block creation. 
Users have different roles and have a incentive to behave honestly in their role.

% The decentralized incentive structure is the core of the blockchain.

%In DPoS, delegators can also propose changes, which are to be voted by network. Transaction validation and block creation are performed by witnesses. Leader selection is deterministic and in a round-robin fashion. Witnesses are paid in terms of block generation and witnesses failing to do so are denied their privileges. 

%{\bf Voting} 
In DPoS systems, users can \emph{vote} to select \emph{witnesses}, to whom they trust, to validate transactions. For top tier witnesses that have earned most of the votes, they earn the right to validate transactions. Further, users can also \emph{delegate} their voting power to other users, whom they trust, to vote for witnesses on their behalf. 
In DPoS, votes are weighted based on the stake of each voter. A user with a small stake can become a top tier witness, if it receives votes from users with large stakes.

%{\bf Witnesses.} 

Top witnesses are responsible for validating transactions and creating blocks, and then get fee rewards. Witnesses in the top tier can exclude certain transactions into the next block. But they cannot change the details of any transaction. There are a limited number of witnesses in the system.
A user can replace a top tier witness if s/he gets more votes or is more trusted. Users can also vote to remove a top tier witness who has lost their trust. Thus, the potential loss of income and reputation is the main incentive against malicious behavior in DPoS.

%{\bf Delegates.} 
Users in DPoS systems also vote for a group of \emph{delegates}, who are trusted parties responsible for maintaining the network. The delegates are in charge of the governance and performance of the entire blockchain protocol. But the delegates cannot do transaction validation and block generation.
For example, they can propose to change block size, or the reward a witness can earn from validating a block. 
The proposed changes will be voted by the system's users.

DPoS brings various benefits:
(1) faster than traditional PoW and PoS Stake systems; 
(2) enhance security and integrity of the blockchains as each user has an incentive to perform their role honestly.
(3) normal hardware is sufficient to join the network and become a user, witness, or delegate.
(4) more energy efficient than PoW.

{\bf Leasing Proof Of Stake} Another type of widely known PoS is Leasing Proof Of Stake (LPoS). Like DPoS, LPoS allows users to vote for a delegate that will maintain the integrity of the system. Further, users in a LPoS system can lease out their coins and share the rewards gained by validating a block.

There are a number of surveys that give comprehensive details of PoW and PoS, such as~\cite{sheikh2018proof, panarello2018survey}. There are other successors of PoS, such as \emph{Proof of Authority} (PoA). In PoA, the reputation of the validator acts as the stake~\cite{ProofofAuth}. PoA system is a permissioned system, which is controlled by a set of validators. Higher throughput can be achieved by reducing the number of messages sent between the validators. Reputation is difficult to regain once lost and thus is a better choice for ``stake''.

\subsection{DAG-based approaches}

DAG-based approaches have currently emerged as a promising alternative to the PoW and PoS blockchains.
The notion of a \emph{DAG} (directed acyclic graph) was first coined in 2015 by DagCoin~\cite{dagcoin15}. Since then, DAG technology has been adopted in numerous systems, for example, ~\cite{dagcoin15, sompolinsky2016spectre, PHANTOM08, PARSEC18, conflux18}. Unlike a blockchain, DAG-based system facilitate consensus while achieving horizontal scalability. This section will present the popular DAG-based approaches.

%\subsubsection{Tangle}
Tangle is a DAG-based approach proposed by IOTA~\cite{tangle17}. 
Tangle uses PoW to defend against sybil and spam attacks. Good actors need to spend a considerable amount of computational power, but a bad actor has to spend increasing amounts of power for diminishing returns. Tips based on transaction weights are used to address the double spending and parasite attack. 

%\subsubsection{Byteball}
Byteball~\cite{byteball16} introduces an internal pay system called Bytes used in distributed database. Each storage unit is linked to previous earlier storage units. The consensus ordering is computed from a single main chain consisting of roots. Double spends are detected by a majority of roots in the chain.

%\subsubsection{Hashgraph}
Hashgraph~\cite{hashgraph16} introduces an asynchronous DAG-based approach in which each block is connected with its own ancestor. Nodes randomly communicate with each other about their known events. Famous blocks are computed by using \textit{see} and \text{strong see} relationship at each round. Block consensus is achieved with the quorum of more than 2/3 of the nodes.

%\subsubsection{RaiBlocks}
RaiBlocks~\cite{raiblock17} was proposed to improve high fees and slow transaction processing. Consensus is obtained through the balance weighted vote on conflicting transactions. Each participating node manages its local data history. Block generation is carried similarly as the anti-spam tool of PoW. The protocol requires verification of the entire history of transactions when a new block is added.

Phantom~\cite{PHANTOM08} is a PoW based permissionless protocol that generalizes Nakamoto’s blockchain to a DAG. A parameter $k$ is used to adjust the tolerance level of the protocol to blocks that were created concurrently. The adjustment can accommodate higher throughput; thus avoids the security-scalability tradeoff as in Satoshi’s protocol. A greedy algorithm is used on the DAG to distinguish between blocks by honest nodes and the others. It allows a robust total order of the blocks that is eventually agreed upon by all honest nodes.

Like PHANTOM, the GHOSTDAG protocol selects a $k$-cluster, which induces a colouring of the blocks as Blues (blocks in the selected cluster) and Reds (blocks outside the cluster). GHOSTDAG finds a cluster using a greedy algorithm, rather than looking for the largest $k$-cluster.

Spectre~\cite{sompolinsky2016spectre} uses DAG in a PoW-based protocol to tolerate from attacks with up to 50\% of the computational power. The protocol gives a high throughput and fast confirmation time. Sprectre protocol satisfies weaker properties in which the order between any two transactions can be decided from the transactions by honest users; whilst conventionally the order must be decided by all non-corrupt nodes. 

Conflux~\cite{conflux18} is a DAG-based Nakamoto consensus protocol. It optimistically processes concurrent blocks without discarding any forks. The  protocol achieves consensus on a total order of the blocks, which is decided by all participants. Conflux can tolerate up to half of the network as malicious while the BFT-based approaches can only tolerate up to one third of malicious nodes.

Parsec~\cite{PARSEC18} proposes a consensus algorithm in a randomly synchronous BFT network. It has no leaders, no round robin, no PoW and reaches eventual consensus with probability one. Parsec can reach consensus quicker than Hashgraph~\cite{hashgraph16}. The algorithm reaches 1/3-BFT consensus with very weak synchrony assumptions. Messages are delivered with random delays, with a finite delay in average.

Blockmania~\cite{Blockmania18}
achieves consensus with several advantages over the traditional pBFT protocol. In Blockmania, nodes in a quorum only emit blocks linking to other blocks, irrespective of the consensus state machine. The resulting DAG of blocks is used to ensure consensus safety, finality and liveliness.
It reduces the communication complexity to $O(N^2)$ even in the worse case, as compared to pBFT's complexity of $O(N^4)$.

In this paper, our \stakedag\ protocol is different from the previous work. We propose a general model of DAG-based consensus protocols, which uses Proof of Stake for asynchronous permissionless BFT systems. \stakedag\ protocol is based on our previous DAG-based protocols\cite{lachesis01,fantom18} to achieve asynchronous non-deterministic pBFT. The new $S_\phi$ protocol, which is based on our ONLAY framework~\cite{onlay19}, uses graph layering to achieve scalable, reliable consensus in a leaderless aBFT DAG.

%% file: stakedagmodel.tex
\newpage

\section{A General Model of \stakedag\ Protocols}\label{se:model}

In this section, we present a general model of \stakedag\ protocols that are Proof of Stake DAG-based for trustless systems.
In this general model of a stake-based consensus protocol, the stakes of the participants are paramount to a trust indicator.

A stake-based consensus protocol consists of nodes that vary in their amount of stakes. Presumingly,
any participant can join the network. Participants can increase their impact as well as their contributions over the Fantom network by the amount of stakes or FTM tokens they possess.
For those who obtain more FTM tokens, they have more validating power to check for the validity of blocks and the underlying DAG.

We then introduce the key concepts of our \stakedag\ protocol as follows.

\subsection{Stake}
Each participant node of the system has an account.
The \emph{stake} of a participant is defined based on their account balance.  The account balance is the number of tokens that was purchased, accumulated and/or delegated from other account.
Each participant has an amount of stakes $w_i$, which is a non-negative integer.

Each participant with a positive balance can join as part of the system. A user has a stake of 1, whereas a validator has a stake of $w_i > 1$. 
The number of stakes is the number of tokens that they can prove that they possess. 
In our model, a participant, whose has a zero balance, cannot participate in the network, either for block creation nor block validation.

\subsection{Validating Power and Validator Types}

In our previous papers~\cite{lachesis01,fantom18},  a family of Lachesis protocols was proposed. In Lachesis, every node can submit new transactions, which are batched in a new event block, and it can communicate its new (own) event blocks with peers. Blocks are then validated by all nodes, all of which have the same validating power.

Here, in this paper, we present a general model that distinguishes \stakedag\ participants by their validating power.
In this model, every node $n_i$ in $\mathfrak{S}$ protocols has a stake value of $w_i$. This value is then used to determine their trust or validating power in validation of event blocks.

Every node in $\mathfrak{S}$ protocols can create new event blocks and can communicate them with all other nodes. All the nodes can validate the event blocks. Their validating power is determined by their stake $w_i$.

Ideally, we could allow participants can do validation only without the need to participate into block creation, if they wish. However, for an asynchronous DAG-based consensus, validators-only in \stakedag\ protocol still need to create event block to indicate that which blocks (and all of its ancestors) have been validated.

There are two general types of validators that we are considered in this paper. They are described as follows.
\begin{itemize}
	\item{\bf{Validators as power nodes:} } {Every node, either user or validator, can create and validate event blocks. 
		Both users and validators can submit new transactions via their new event blocks.
		
		The only difference between users and validators of this type is that a validator has more validating power of, say $w_i > 1$, while a user has a validating power of 1.}
	\item{\bf{Validators as validation-only nodes:} } {We now consider another type of validators. Users and validators of this type play more specific roles. 
		A user node can create new event blocks, which contain transactions. User can validate any event blocks.
		
		In this type, validators are validation-only nodes. Validators can create \emph{empty} event blocks, but the event blocks created by validator-only nodes contain no transactions. Validators can validate event blocks.
		}
\end{itemize}

For stake-based consensus, the difference between the two validator types does not have any significant impact on consensus. The two types are different with respect to the responsibility (validation) and the economic expectation of a participant who joins the \stakedag\ protocol.
In particular, the first type allows any validator to include new transactions in the new event blocks it creates, while the second type prohibits validators from adding new transactions.

Remarkably, the resulting OPERA chains on the nodes are the same for two types of validators. Thus, the computation of consensus does not change between these two types. The two models of validators are considered as implementation-specific.
In fact, the first type is more general than the second type.

For the sake of simplicity, we distinguish two types of nodes: users and validators.  We consider only the first type of validators, throughout the rest of the paper.

\subsection{Event Block Creation}

In order to create a new event block, a node can choose a top event block(s) of another node as a parent. In \stakedag\ protocol, users and validators can choose a top event block(s) of another  node/validator node to as an other-parent reference to the new event block.
Prior to creating a new event block, the node will first validate its current block and the selected top event block(s) from other node(s).

In a \stakedag\ protocol, each event block has $k$ references to other event blocks using their hash values. The references of a new event block satisfy the following conditions:

\begin{enumerate}
	\item Each of the $k$ referenced event blocks is the top event block of its own node.
	\item One of the references must be a self-ref (or self-parent)  that references to an event block of the same node. 
	\item The other $k$-1 references, which are also known as other-parent or other-ref references, refer to $k$-1 top event blocks on other nodes.
\end{enumerate}

\subsection{S-OPERA chain: A Weighted DAG}

In \stakedag\ protocol, a (participant) node is a server (machine) of the distributed system.
Each node can create messages, send messages to and receive messages from other nodes. The communication between nodes is asynchronous.

\begin{figure}[ht]
	\centering
	\includegraphics[width=0.4\linewidth]{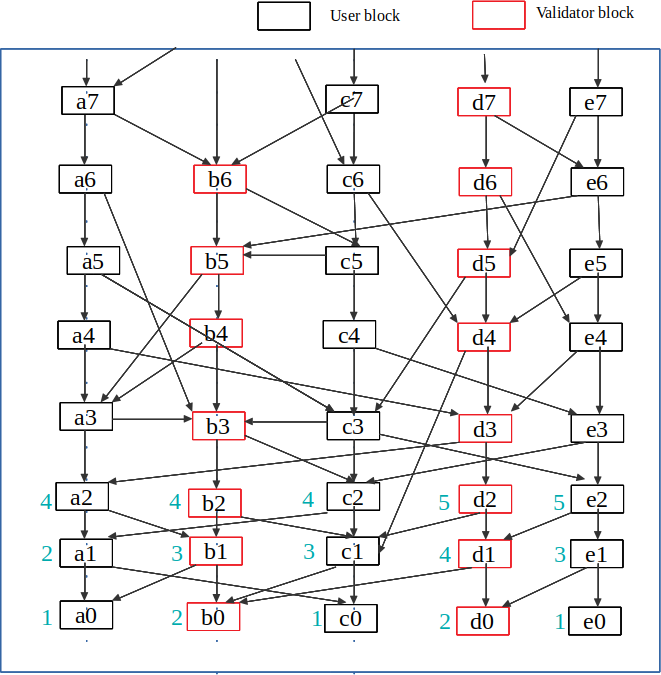}
	\caption{An example of an S-OPERA chain in \stakedag. The validation scores of some selected blocks are shown.}
	\label{fig:stakedag-validatorscore}
\end{figure}
Figure \ref{fig:stakedag-validatorscore}
depicts an S-OPERA chain obtained from the DAG.
In this example, validators have validating power of 2, while users have validating power of 1. Blocks created by validators are highlighted in red color. First few event blocks are marked with their validating power. Leaf event blocks are special as each of them has a validation score equal to their creator's validating power.

The core idea of \stakedag\ protocol is to use a DAG-based structure, namely \emph{S-OPERA chain}, which is based on the concept of OPERA chain in our Lachesis protocol~\cite{lachesis01}. 
S-OPERA chain is a weighted directed acyclic graph $G$=($V$,$E$), where $V$ is the set of event blocks, $E$ is the set of edges between the event blocks. Each vertex (or event block) is associated with a \emph{validation score}. Each event block also has a validating score, which is the total weights of the roots reachable from it. 
When an block becomes a root, it is assigned a \emph{weight}, which is the same with the validating power of the creator node.

Let $\mathcal{W}$ be the total validating power of all nodes. 
For consensus, the algorithm examines whether an event block has a validation score of at least $2\mathcal{W}/3$. A validation score of $2\mathcal{W}/3$ means the event block has been validated by more than two-thirds of total validating power in the S-OPERA chain. 

\subsection{Stake-based Validation}
In \stakedag, we use the following validation process. 
A node, when validating blocks, never assumes the honesty or dishonesty of any block creator. Instead it verifies the blocks just like other consensus engines operating in an untrusted environment.
A node must validate its current block and the received ones before it attempts to create or add a new event block into its local DAG. A node must validate its (own) new event block before it communicates the new block to other nodes.

Note that, \stakedag\ consensus protocols has a more general model than that of the Lachesis protocols~\cite{lachesis01}. All nodes in Lachesis can be considered as 'user' nodes in \stakedag, since the validating power is 1 by default.

A root is an important block that is used to compute the final consensus of the event blocks of the DAG.
In Lachesis, when a block can reach more than 2/3 of the roots of the network, it becomes a root.  Unlike in Lachesis protocol, \stakedag\ procol takes the stakes of participants into account to compute the consensus of blocks. A stake number of a block is the sum of the validating power of the nodes whose roots can be reached from the block. When a block receives more than 2/3 of the entire validating power of the network, it becomes a root.

\subsection{Validation Score}

When an event block can reach a root, it takes the validating power of the root's creator. The validation score of a block is the sum of accummulated validating powers that the block has gained.
The validation score of a block $v_i$ $\in$ $G$ is denoted by $s(v_i)$.

The \emph{validation score} of a block is used to determine whether the block is a new root. If the block's score is greater than 2/3 of the total validating power, the block becomes a root.
The weight of a root $r_i$ is denoted by $w(r_i)$, which is the weight $w_j$ of the creator node $j$ of the $r_i$.

\subsection{Flagtable Calculation}
To quickly compute the validation score of event blocks, we introduce stake-based flagtable. Each flagtable is a mapping of a root and the stake associated with the creator of that root.

Each event block has a flagtable that stores a set of the roots that are reachable from the block. For flagtable, we only consider the roots in the current active root set.

The validation score of a block is computed as the sum of all validating powers of all the roots contained in the block's flagtable.

\begin{figure}[ht]
	\centering
	\includegraphics[width=0.9\linewidth]{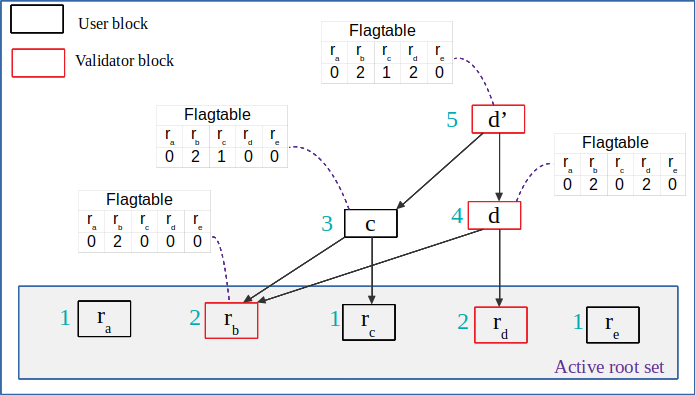}
	\caption{An Example of Flagtable}
	\label{fig:flagtable}
\end{figure}

Figure~\ref{fig:flagtable} shows an example of flagtable calculation for event blocks. The current active root set contains five roots $r_a$, $r_b$, $r_c$, $r_d$ and $r_e$ with their validating powers of 1, 2, 1, 2, 1, respectively. The flagtable of block $r_b$ contains a single mapping entry ($r_b$, 2), and thus the validation score of $r_b$ is the same with the validating power of node $b$, which is 2. The validating power of block $c$ is the sum of 2 and 1, which gives a result of 3. Similarly, the validation scores of blocks $d$ and $d'$ are 4 and 5, respectively.

\subsection{Root Selection}
Figure~\ref{fig:rootselection} depicts the process of selecting a new root in \stakedag. A node can create a new event block or receive a new block from other nodes. When a new event block is created, its flagtable is updated. Like in Lachesis protocol, each block has a flagtable that contains the sets of roots that are reachable from the block.

\begin{figure}[ht]
	\centering
	\includegraphics[width=\linewidth]{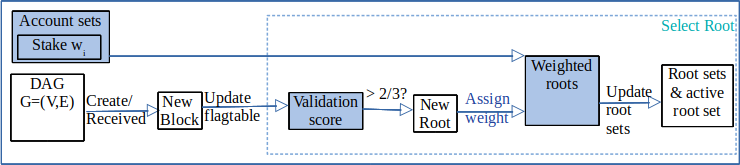}
	\caption{The steps for root selection in \stakedag}
	\label{fig:rootselection}
\end{figure}

Root events get consensus when 2/3 validating power is reached.
When a new root is found, 'Assign weight' step will assign a weight to the root.
The root's weight is set to be equal to  the validating power $w_i$ of its creator.
The validation score of a root is unchanged.

\subsection{Main-chain}

For faster consensus, we introduce the \emph{Main-chain}, which is a special sub-graph of the S-OPERA chain. 
The Main chain is an append-only list of blocks, that caches the final consensus ordering of the finalized Atropos blocks.
The local hashing chain is useful to improve path search to quickly determine the closest root to an event block. Root event blocks are important blocks that reach $2/3$ of the network power.  
After the topological ordering is computed over all event block, Atropos blocks are determined and form the Main chain. 

Each participant has an own copy of the Main chain and can search consensus position of its own event blocks from the nearest Atropos.
The chain provides quick access to the previous transaction history to efficiently process new coming event blocks. 
With the Main chain, unknown participants or attackers can be easily detected.

%% file: staking.tex
\newpage
\section{\stakedag\ Staking model}\label{se:staking}

This section presents our staking model for PoS DAG-based \stakedag\ protocols. We first give general definitions and variables of our staking model, and then present the model in details.

\subsection{Definitions and Variables}

Below are the general definitions and variables that are important in \stakedag\ protocol.

\textbf{General definitions}

\begin{longtable}{p{1cm} p{13cm}}
		$\mathcal{F}$ & denotes the network itself \\[4pt]
		$N$  &    the set of nodes in the network \\[4pt]
		$SPV$          &"Special Purpose Vehicle" - a special smart contract acting as an internal market-maker for FTG tokens, managing the collection of transaction fees and the payment of all rewards \\[4pt]
		$FTM$          & main network token \\[4pt]
		$FTG$          & network transaction token (gas) \\[4pt]
		$E$            & set of event blocks in $\mathcal{F}$ \\[4pt]
		$E(d)$         & set of all event blocks validated on day $d$, their number being $|E(d)|$
\end{longtable}

\textbf{Accounts}

\begin{longtable}{p{1cm} p{1.5cm} p{11cm}}
		$\mathbb{U}$  & & set of all participant accounts in the network \\[4pt]
		$\mathbb{A}$  & $\subset \mathbb{U}$ & accounts with a positive FTM token balance \\[4pt]
		$\mathbb{A}_G$  & $\subset \mathbb{U}$ & accounts with a positive FTG token balance \\[4pt]
		$\mathbb{S}$  & $\subseteq \mathbb{A}$ & accounts that have staked for validation (some of which may not actually be validators) \\[4pt]
		$\mathbb{V}$  & $\subseteq \mathbb{S}$ & validating accounts, corresponding to the set of the network's validating nodes
\end{longtable}

A participant with an account having a positive FTM token balance, say $i \in \mathbb{A}$, can join the network.
But an account $i$ in $\mathbb{S}$ may not participate in the protocol yet. Those who join the network belong to the set $\mathbb{V}$. 

\dfnn{Validating account}{A validating account is an account with a positive FTM token balance}.

Note that, the set of validating accounts $\mathbb{V}$ contain both users and validators, as presented in our general model in Section~\ref{se:model}.

\textbf{Network parameters subject to on-chain governance decisions}

\begin{longtable}{p{1cm} p{1.5cm} p{11cm}}
	$F$  & $3.175e9$ & total supply of FTM tokens \\[4pt]
	$\delta$  & $30$ & period in days for determining Proof of Importance \\[4pt]
	$\lambda$ & $90$ & period in days after which validator staking must be renewed, to ensure activity \\[4pt]
	$\varepsilon$ & 1 & minimum number of tokens that can be staked by an account for any purpose \\[4pt]
	$\theta$  & 30\% & impact of Proof of Importance for rewards \\[4pt]
	$\xi$  & 50\% & impact of Proof of Importance for transaction staking \\[4pt]
	$\phi$    & 30\% & SPV commission on transaction fees \\[4pt]
	$\mu$     & 15\% & validator commission on delegated tokens
\end{longtable}

\textbf{Tokens held and staked}

Unless otherwise specified, any mention of \textit{tokens} refers to FTM tokens.

\dfnn{Token helding}{The token helding  $t_i$ of an account is number of FTM tokens held by account $i \in \mathbb{A}$.}

\begin{longtable}{p{1cm} p{1.5cm} p{11cm}}
	$t_i$        &                 & number of FTM tokens held by account $i \in \mathbb{A}$ \\[4pt]
	$t_i^{[x]}$      & $> \varepsilon$ & transaction-staked tokens by account $i$\\[4pt]
	$t_i^{[d]}(s)$   & $> \varepsilon$ & tokens delegated by account $i$ to account $s \in \mathbb{S}$ \\[4pt]
	$t^{[d]}(s)$     &                 & total of tokens delegated to account $s \in \mathbb{S}$ \\[4pt]
	$t_i^{[d]}$      &                 & total of tokens delegated by account $i$ to accounts in $\mathbb{S}$ \\[4pt]
	$t_i^{[s]}$      &                 & validation-staked tokens by account $i$ \\[4pt]
	$T^{[s]}_{min}$  & 0.1\%           & minimum tokens staked by $v \in \mathbb{V}$, as a percentage of $F$  \\[4pt]
	$T^{[s]}_{max}$  & 0.4\%           & maximum tokens staked by $v \in \mathbb{V}$, as a percentage of $F$ \\[4pt]
	$M$          & 15              & delegation multiplier - maximum ratio of delegated versus staked tokens\\[4pt]
	$Q$          &              & multiplier - maximum fraction of staked or delegated tokens over the holding
\end{longtable}

The sum of tokens staked or delegated by an account $i \in \mathbb{A}$ cannot exceed the amount of tokens held:
\begin{equation}
t_i^{[x]} + t_i^{[s]} + \sum_{s \in \mathbb{S}} t_i^{[d]}(s)  \leq t_i
\end{equation}

The following limits apply for token staked for validation by $s \in \mathbb{S}$:
\begin{equation*}
T^{[s]}_{min} \cdot F \leq t_s^{[s]} \leq T^{[s]}_{max} \cdot F
\end{equation*}

The total amount of tokens delegated to an account $s \in \mathbb{S}$ is:
\begin{equation}
t^{[d]}(s) = \sum_{i \in \mathbb{A}} t_i^{[d]}(s)
\end{equation}

The total amount of tokens delegated by an account $i \in \mathbb{A}$ is:
\begin{equation}
t_i^{[d]} = \sum_{s \in \mathbb{S}} t_i^{[d]}(s)
\end{equation}

The sum of tokens delegated to an account $s \in \mathbb{S}$ cannot exceed a fixed multiple of tokens staked by that account:

\begin{equation*}
t^{[d]}(s) \leq M \cdot t_s^{[s]}
\end{equation*}

Finally, there is a limit on the amount of tokens that can be transaction-staked or delegated:

\begin{equation*}
t_i^{[d]} \leq Q \cdot t_i \qquad \quad t_i^{[x]} \leq Q \cdot t_i
\end{equation*}

\subsection{Absolute Weight Model}

\dfnn{Weight in $\mathcal{F}$}{
The weight of an account $i \in \mathbb{A}$ is equal to its token holding $t_i$.}

\dfnn{Importance in $\mathcal{F}$}{
The importance is proportional to gas use in the overall network over the most recent period of $\delta$ days.}

The importance of a node $i$ is computed by:
\begin{equation}
\hat{g}_i = \frac{g_i}{\mathcal{G}} \mathcal{P},
\end{equation}
where 
\begin{longtable}{p{1cm} p{13cm}}
	$g_i$          &  gas used by account $i\in \mathbb{U}$ during past $\delta$ days \\
	$\mathcal{G}$  &  gas used in the entire network in past $\delta$ days. So $\mathcal{G} = \sum_{i \in \mathbb{U}} g_i$.\\
	$\hat{g}_i$    &  importance of $i \in \mathbb{U}$, rebased to be comparable with $t_i$ \\
	$\mathcal{P}$ & 	the sum of all the network's transacting power \\
\end{longtable}

\dfnn{Transacting Power}{
The transacting power of an account is defined as a weighted average of an account's weight and importance in $\mathcal{F}$.}

\begin{longtable}{p{1cm} p{13cm}}
	$x_i$  &  transacting power of account $i \in \mathbb{U}$ \\[4pt]
	$X$    &  total transacting power of $\mathbb{U}$. Thus, $X = \sum_{i \in \mathbb{U}} x_i = \mathcal{P}$.
\end{longtable}

Transacting power of an account is computed by:
\begin{equation}
x_i = \xi \hat{g}_i + (1 - \xi) t_i
\end{equation}

Note that, the above model of transacting power includes a gas factor used in smart contracts. Without gas factor, the transacting power of an account $i \in \mathbb{U}$ is simply given by $x_i = t_i$, which is the token helding of the account.

\textbf{Network Performance and Transaction Slots}

\begin{longtable}{p{1cm} p{2.2cm} p{11cm}}
		$\Pi$            & 5,000,000,000   & Estimated maximum network processing power, in FTG per second  \\[4pt]
		$\Theta$         & 500,000 & Estimated maximum network throughput, in Bytes per second  \\[4pt]
		$\sigma^g$       &     & transacting power needed for a slot of 1 FTG per second \\[4pt]
		$\sigma^b$       &     & transacting power needed for a slot of 1 Byte/second of throughput
\end{longtable}

Given that at most $F$ tokens can be transaction-staked, we have:
\begin{equation}
\sigma^g = \frac{F}{\Pi} \approx 0.635 \qquad \qquad
\sigma^b = \frac{F}{\Theta} \approx 6350
\end{equation}

\subsection{Relative Weight Model}

Since only a fraction of tokens are staked or delegated for validation, we define a \emph{relative weight} in order to correctly determine the total validating power of the entire network $\mathcal{F}$.

\dfnn{Weight in $\mathbb{V}$}{
A relative weight of a validator $v$ is a relative value of the validating power of $v$ computed from the total validating power of the entire network $\mathcal{F}$.}

\begin{longtable}{p{1cm} p{14cm}}
	$w_v$ & tokens staked by, and delegated, to validator $v \in \mathbb{V}$, which represents the weight of this validator \\[4pt]
	$\mathcal{W}$  &  total tokens staked by, and delegated to, validators. That is, $\mathcal{W} = \sum_{v \in \mathbb{V}} w_v$.
\end{longtable}

The relative weight of $v$ is computed by:
\begin{equation}
\quad  w_v = t_v^{[s]} + \sum_{i \in \mathbb{A}} t_i^{[d]}(v) = t_v^{[s]} + t_i(v)
\end{equation}

\dfnn{Importance in $\mathbb{V}$}{
The importance in $\mathbb{V}$ is proportional to gas use by accounts that have staked or delegated tokens.}

\begin{equation}
\hat{h}_v = \frac{h_v}{\mathcal{H}} \mathcal{W},
\end{equation}
where

\begin{longtable}{p{1cm} p{13cm}}
	$h_v$         &  gas use over the past $\delta$ days attributable to validator $v \in \mathbb{V}$. That is, $h_v = g_v + \sum_{i \in \mathbb{A}} g_i\frac{t_i^{[d]}(v)}{t_i}$. \\[4pt]
	$\mathcal{H}$ &  total gas use over the past $\delta$ days attributable to all validators. So $\mathcal{H} = \sum_{v \in \mathbb{V}} h_v$.
	\\[4pt]
	$\hat{h}_v$   &  importance of $v \in \mathbb{V}$, rebased to be comparable with $w_v$
\end{longtable}

\subsection{Validating power}

We first present the simplest model of validating power, which is defined as the number of tokens held by an account.

\dfnn{Validating power - Simple}{
	The validating power of a validator $v \in \mathbb{V}$ is defined as validator's weight. The weight of a validator can be the token helding $t_i$, or the validator's weight $w_v$.}

We then also present a general model of the validating power, as follows.

\dfnn{Validator Score}{
The validator score is the score given for each validator $v \in \mathbb{V}$.}

\begin{longtable}{p{1cm} p{1.5cm} p{11cm}}
	$s_v$    &  $\in [0, 1]$  & the total score of validator $v \in \mathbb{V}$
\end{longtable}

For each validator $v \in \mathbb{V}$, validator score will be a number between $0$ and $1$ representing a composite of three different performance metrics.

\begin{longtable}{p{1cm} p{1.5cm} p{11cm}}
	$s_v^{[w]}$  &  $\in [0, 1]$  & validating performance score of $v$ \\[4pt]
	$s_v^{[t]}$  &  $\in [0, 1]$  & transaction origination score of $v$ \\[4pt]
	$s_v^{[p]}$  &  $\in [0, 1]$  & processing power score of $v$
\end{longtable}

\dfnn{Validating power}{
The validating power of a validator is defined as a weighted average of a validator's weight and importance, multiplied by its score.}

Thus, the validating power of a node is computed by:
\begin{equation}\label{eq:valpower}
p_v = s_v [ \theta \cdot \hat{h}_v + (1 - \theta) \cdot w_v ],
\end{equation}
where

\begin{longtable}{p{1cm} p{13cm}}
	$p_v$  &  validating power of validator $v \in \mathbb{V}$ \\[4pt]
	$P$    &  total validating power of $\mathbb{V}$. So $P = \sum_{v \in \mathbb{V}} p_v$.
\end{longtable}

Note that, Equation~\ref{eq:valpower} gives a general model for calculating validating power, which includes the notion of gas-based importance $h_v$. The gas is used for smart contracts. When gas factor is not included, the validating power of an account $v \in \mathbb{V}$ is simply given by $p_v = w_v$, which is the relative weight of $v$.

\subsection{Block consensus and Rewards}
We then present notations and definitions in our model of rewards.

\dfnn{Validation score} {Validation score of a block is the total validating power that a given block can achieve from the validators $v \in \mathbb{V}$.}

\dfnn{Validating threshold}{Validating threshold is defined by $2/3$ of the validating power that is needed to confirm an event block to reach consensus.}

Below are the variables that define the block rewards and their contributions for participants in Fantom network.

\begin{longtable}{p{1.2cm} p{2.5cm} p{11cm}}
	$Z$      & 996,341,176 & total available block rewards of $FTM$, for distribution by the $SPV$ during the first 1460 days after mainnet launch\\[4pt]
	$F_s$    &   & $FTM$ tokens held by the $SPV$ \\[4pt]
	$F_c$    & $F - F_s$ & total circulating supply \\[4pt]
	$B(d)$   &   & total transaction fees paid by network users on day $d$ \\[4pt]
	$R^{[v]}(d)$ & $\phi * B(d)$  & total transaction rewards retained by the $SPV$ on day $d$ \\[4pt]
	$R^{[x]}(d)$ & $(1-\phi)*B(d)$  & total transaction rewards to be distributed on day $d$ \\[4pt]
	$R^{[b]}(d)$ &   & total block rewards to be distributed on day $d$ \\[4pt]
	$R(d)$   & $R^{[x]}(d) + R^{[b]}(d)$  & total rewards to be distributed on day $d$ \\[4pt]
	$R_v(d)$ & $R(d)*(p_v/P)$  & total daily rewards attributable to validator $v \in \mathbb{V}$ on day $d$, that will be paid out to the validator and to all user accounts who have delegated tokens to that validator \\[4pt]
	$D_i^{[v]}(d)$ &   & total daily delegation rewards received by $i \in \mathbb{A}$ attributable to validator $v \in \mathbb{V}$ on day $d$ (exclusive of validator rewards) \\[4pt]
	$D_i(d)$ & $\sum_{v \in \mathbb{V}} D_i^{[v]}(d)$  & total delegation rewards received by $i$ from all validators on day $d$ \\[4pt]
	$D^{[v]}(d)$ & $\sum_{i \in \mathbb{A}} D_i^{[v]}(d)$  & total delegation rewards received by all delegators attributable to validator $v$ \\[4pt]
	$I_v(d)$ &   & daily validator rewards received by $v \in \mathbb{V}$ on day $d$ (exclusive of delegation rewards)
\end{longtable}

Block rewards will be distributed over 1460 days (4 years less one day) after launch, corresponding to $Z / 1460$ per day during that period:

\begin{equation}
R_b(d)=\begin{cases}
$682,425.46$, & \text{during 1460 days after mainnet launch}.\\
0,          & \text{otherwise}.
\end{cases}
\end{equation}

%%%%%
%%%%% Overview of Token Staking and Delegation
%%%%%

\subsection{Token Staking and Delegation}\label{se:stakingoverview}

FTM tokens will have multiple uses in the Fantom system. Participants can chose to stake or delegate tokens into their accounts. When staking or delegating, the validating power of a node is based on the number of FTM tokens held, plus other factors such as the “importance” to the network as measured by the gas (described in previous section).

We describe three potential ways for staking to achieve wealth.

\dfnn{Transaction staking}{Participants can gain more stakes or tokens via the transaction staking.}
Transaction submitters will gain transaction fees if the transactions are successfully validated and reach finality.
This style of staking helps increases transaction volume on the network. The more liquidity, the more transaction staking they can gain.

\dfnn{Validation staking}{By validating blocks, a participant can gain validation rewards. Honest participants can gain block rewards for  successfully validated blocks.}

With validation staking, one can join the network using a moderate hardware and a certain small amount of tokens. Once participating the network, the participant can achieve block rewards for blocks that they co-validated and gain transaction fees from transaction submitters for the successfully finalized transactions. The more stake they gain, the validating power they will have and thus the more rewards they can receive as it is proportional to their validating power.

\dfnn{Validation delegation}{Validation deligation allows a participant to deligate all or part of their tokens to another participant(s). Deligating participants can gain a share of block rewards and transaction fees, based on the amount of delegated stake.}

Our model of delegation of stakes between participants allow various coordination amongst participants. Early stage participants can borrow stake from other participants to gain more rewards. Meanwhile, participants with large amount of stakes can deligate their stakes to another participant while still earning some shared rewards.

%In addition, FTM holders will also be able to participate in on-chain voting, but that is outside the scope of this paper.

%%%%%
%%%%% Validation Staking
%%%%%

\subsubsection{Delegating Staking}

Stakeholders will be able delegate a portion of their tokens to validating nodes. Validators will not be able to spend delegated tokens, which will remain secured in the stakeholder’s own address. Validators will receive a fixed proportion of the validator fees attributable to delegators.

Participants are compared by their performance. Higher performing node, with high uptimes and successful validation rates, will earn more rewards. Delegators will be incentivised to choose nodes that have a high validator score, i.e. are honest and high performing.

Delegators can delegate their tokens for a maximum period of days, after which they need to re-delegate. The requirements to delegate are minimal:   

\begin{itemize}
	\item \textit{Security deposit}: None
	\item \textit{Minimum number of tokens to delegate}: 1
	\item \textit{Minimum lock period}: None
	\item \textit{Maximum number of validators a user can delegate to}: None
	\item \textit{Maximum number of tokens that can be delegated to a validator}: 15 times the number of tokens the validator is staking
\end{itemize}

%Fantom aims to promote and develop software that makes it convenient to delegate to validators through a straightforward UI.

\subsubsection{Transaction-Based Staking}

We present more details of \emph{transaction-based staking} used in \stakedag. With transaction-based staking, FTM holders can stake a portion of their tokens to secure a guaranteed transaction volume on the network.

Staking tokens gives FTM holders a guaranteed transaction slot, which consists of the following:
\begin{itemize}
	\item gas: expressed in FTG/second
	\item data throughput: expressed in Bytes/second
\end{itemize}

The size of the issued slot will be proportional to the transacting power of the tokens staked.

{\bf Network Processing Power and Throughput}  Here, we give an estimate of the maximum gas that can be spent per second. The gas usage depends heavily on the processing power. At the year of this writing, the best desktop processors (such as the Intel Core i9 Extreme Edition) have already reached teraflop speeds - one trillion floating point operations per second. 

Every instruction processed by the Fantom Virtual Machine ("FVM") will carry some overhead, as it has to verify signatures and track gas spent. The below gives an estimate:
\begin{itemize}
	\item Assume that only 10\% of processing power is available for executing transactions and smart contracts. Also assume that a single multiplication by the FVM costs 100 floating-point operations.
	\item FVM can process one billion multiplications per second (i.e 5 billion gas - utilising Ethereum's gas pricing as a guide).
	\item $5,000,000,000 / 21,000 = 238,095$ basic transactions, which are simple transfers of value from one address to another.
	\item A basic transaction has, on average, a size of 120 bytes.
	\item This would result in a data volume of $230,095 * 120$ = 27 MB / second.
\end{itemize}

From the above, a data volume of 27 MB per second is for new transactions. When consensus protocol traffic is taken into account, it becomes quite too high. This means that if the majority of transactions are relatively simple, the main bottleneck will be network throughput.

The current assumptions are:

\begin{itemize}
	\item $\mathcal{F}$ can support 5 billion gas / second.
	\item A reasonably high maximum transaction throughput of 0.5MB / second.
\end{itemize}

{\bf Staking}
Now we show an example of staking. Assume that each basic transaction has a size of 120 Bytes on average.
In order to be guaranteed one basic transaction per second, a user would need to stake tokens with a transacting power corresponding $120 * 6,350 = 762,000$.
If that user's gas usage is in line with his token holdings, his transacting power will be roughly equal to his token holdings. Thus, the necessary number of transaction-staked FTM tokens will also be approximately 762,000.
This corresponds to $762,000 / 0.635 \approx 1,200,000$ FTG per second.

%% file: Sprotocol.tex
\newpage
\section{Layering-based \stakedag: $S_\phi$ Protocol}\label{se:Sprotocol}

In this section, we present a specific PoS DAG-based protocol, namely $S_\phi$, of the family $\mathfrak{S}$. 
Our general model of \stakedag\ protocols is presented in Section~\ref{se:model}. 
$S_\phi$ protocol is a layering-based approach, which is based on our ONLAY framework~\cite{onlay19}. $S_\phi$ provides a PoS DAG-based solution to build scalable asynchronous distributed ledgers.

\subsection{Framework}
We first give an overal design of the $S_\phi$ protocol. 
Figure~\ref{fig:stakedagsprotocol} shows the overall framework of layering-based \stakedag\ protocol.
Like the general model in Figure~\ref{fig:stakedagframework}, $S_\phi$ protocol includes the steps to update flagtable and validation score for each new event block. Once a block receives more than 2/3 of the validating power of the network, the block becomes a root.

\begin{figure}[htb]
	\centering
	\includegraphics[width=0.95\linewidth]{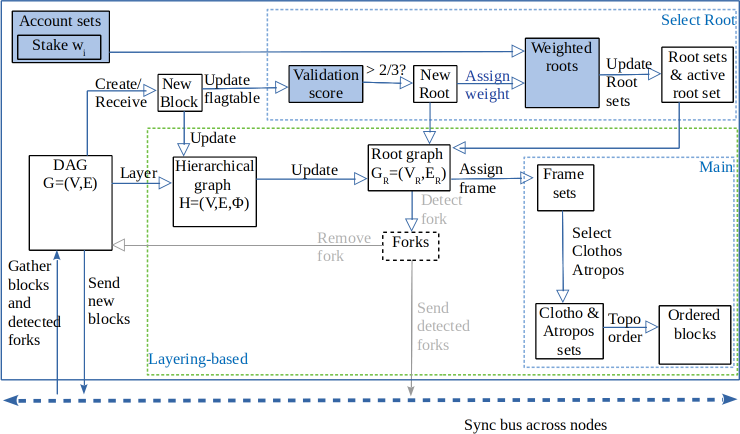}
	\caption{An Overview of Layering-based \stakedag\ Protocol}
	\label{fig:stakedagsprotocol}
\end{figure}

The figure also depicts extra steps used in $S_\phi$ protocol. The main difference with the general model is that $S_\phi$ leverages layering and hierarchical graph obtained from a layering step. Like in ONLAY~\cite{onlay19}, the layer information of the event blocks are used to achieve reliable frame assignment and consensus of finally ordered event blocks. The root graph serves as a convenient data structure to help find new root, detect possible forks and assign new frame to event blocks.

\subsection{Main procedure}
Algorithm~\ref{al:main} shows the main function serving as the entry point to launch \stakedag\ protocol. The main function consists of two main loops, which are they run asynchronously in parallel. The first loop attempts to request for new nodes from other nodes, to create new event blocks and then communicate about them with all other nodes. 
The second loop will accept any incoming sync request from other nodes. The node will retrieve updates from other nodes and will then send responses that consist of its known events.

%main procedure
\begin{algorithm}[H]
\caption{\stakedag\ Main Function}\label{al:main}
\begin{tabular}{ll}
\begin{minipage}{0.5\textwidth}
	\begin{algorithmic}[1]
		\Function{Main Function}{}
		\BState \emph{loop}:
		\State Let $\{n_i\}$ $\leftarrow$ $k$-PeerSelectionAlgo()
		\State Sync request to each node in  $\{n_i\}$
		\State (SyncPeer) all known events to each node in $\{n_i\}$
		\State Create new event block: newBlock($n$, $\{n_i\}$)
		\State (SyncOther) Broadcast out the message
		\State Update DAG $G$
		\State Call ComputeConsensus($G$)
		\BState \emph{loop}:
		\State Accepts sync request from a node
		\State Sync all known events by \stakedag\ protocol
		\EndFunction
	\end{algorithmic} 
\end{minipage}
\begin{minipage}{0.5\textwidth}
	\begin{algorithmic}[1]
		\Function{ComputeConsensus}{DAG}
		\State Apply layering to obtain H-OPERA chain*
		\State Update flagtable
		\State Compute validation score*
		\State Root selection
		\State Compute root graph*
		\State Compute global state
		\State Clotho selection
		\State Atropos selection
		\State Order vertices and assign consensus time
		\State Compute Main chain
		\EndFunction	
	\end{algorithmic}
\end{minipage}
\end{tabular}
\end{algorithm}

In the first loop, a node makes synchronization requests to $k$-1 other nodes to retrieve their top event blocks. It then creates event blocks that references those received event blocks. 
Once created, the blocks are broadcasted to all other nodes. 
In line 3, each node runs the Node Selection Algorithm, to find the next $k$-1 nodes that it will need to communicate with. In line 4 and 5, the node sends synchronization requests to get the latest OPERA chain of these nodes. 
After the latest event blocks are received in the responses, the node creates new event blocks (line 6), and it then broadcasts the created event block to all other nodes (line 7). 
After creating a new event block, the node updates its OPERA chain and then call compute consensus function (line 8 and 9).

The \texttt{ComputeConsensus} function first applies layering on the DAG (line 2). After layering is performed, H-OPERA chain is obtained. For each new event block, it updates the flagtable, which includes roots reachable from the block (line 3). Then from the flagtable, validation score is computed based on the weights of the reachable roots (line 4). It then checks whether the block is a root (line 5) and build/update the root graph (line 6). In line 7, it recomputes a global state based on its local H-OPERA chain. Then the node decides which roots become Clothos (line 8) and which then becomes Atropos (line 9). When new Atropos vertices are confirmed, the algorithm runs a topological sort to get the final ordering for unsorted vertices for final consensus (line 10). Lastly, the main chain is constructed from the Atropos vertices that are newly found and sorted (line 11).

\subsection{Peer selection algorithm}

In order to create a new event block, a node  in \stakedag\ protocol needs 
to synchronize with $k$-1 other nodes for their latest top event blocks. A peer selection algorithm computes a set of $k$-1 nodes that a node should synchronize with.

There are multiple ways to select $k$ - 1 nodes from the set of $n$ nodes. Examples of peer selection algorithms that are mentioned in our previous paper~\cite{onlay19} include:
(1) random peer from $n$ peers;
(2) the least / most used peer(s);
(3) aim for a balanced distribution;
(4) based on some other criteria, such as network latency, successful rates, number of own events.
In general, our protocol does not depend on how peer nodes are selected.

{\bf Stake-based Peer Selection} We also propose new peer selection algorithms, which utilize stakes to select the next peer. Each node has a mapping of the peer $i$ and the frequency $f_i$ showing how many times that peer was selected. These new peer selection algorithms are adapted from the ones above, but take user stakes into account. We give a few examples of the new algorithms, described as follows:
\begin{enumerate}
%[leftmargin=*, listparindent=0.5em, labelsep=0.5em, itemindent=0.5em]
\item {\bf Stake-based selection
}: select a random peer from $n$ peers with a probability proportional to their stakes $w_i$.
\item {\bf Least used selection}: select the peer with the lowest values of $f_i * w_i$.
\item {\bf Most used selection}: select the peer with the highest values of $f_i * w_i$.
\item {\bf Balance selection}: aim for a balanced distribution of selected peers of a node, based on the values $f_i * w_i$.
\end{enumerate}

There are other possible ways to integrate stakes and stake-related criteria into a peer selection algorithm. For example, we can define new algorithms based some other criteria, such as successful validation rates, total rewards, etc.

\subsection{Peer synchronization}

Now we describe the steps to synchronize events between the nodes, as presented in Algorithm~\ref{al:syncevents}.
In the Event synchronization function, each node $n_1$ selects a random peer $n_2$ (from the set of peers computed by peer selection algorithm).  It then sends a sync request consisting of the local known events of $n_1$ to $n_2$. After receiving the sync request from $n_1$, $n_2$ will  compute an event diff with its own known events and will then return the unknown events to $n_1$.
The algorithm assumes that a node always needs the events in topological ordering (specifically in reference to the lamport timestamps). Alternatively, one can simply use a fixed incrementing index or layer assignment to keep track of the top event for each node.

\begin{algorithm}[H]
	\caption{Peer Synchronization}\label{al:syncevents}
\begin{minipage}{0.48\textwidth}
	\begin{algorithmic}[1]
		\Function{sync-events()}{}		
		\State $n_1$ selects a random peer $n_2$ to synchronize with
		\State $n_1$ gets local known events (map[int]int)
		\State $n_1$ sends a RPC Sync request to peer
		\State $n_2$ receives a RPC Sync request
		\State $n_2$ does an EventDiff check on the known event map (map[int]int)
		\State $n_2$ returns the unknown events to $n_1$		
		\EndFunction
	\end{algorithmic}
\end{minipage} \quad
\begin{minipage}{0.475\textwidth}
\begin{algorithmic}[1]
	\Function{sync-stakes()}{}		
	\State $n_1$ selects a random peer $n_2$ to synchronize with
	\State $n_1$ gets stake updates (map[int]int)
	\State $n_1$ sends a RPC Sync request to peer
	\State $n_2$ receives a RPC Sync request
	\State $n_2$ does an StakeDiff check on the peer-stake map (map[int]int)	
	\State $n_2$ returns update stakes to $n_1$		
	\EndFunction
\end{algorithmic}
\end{minipage}
\end{algorithm}

For stake synchronization, the function works in a similar manner.
Each node $n_1$ computes any new stake updates in its local view. It will select a random peer $n_2$ and then sends a sync request consisting of the local stake updates of $n_1$ to $n_2$. After receiving the sync request from $n_1$, $n_2$ will compute a stake diff with its own known stake values and will then return new events to $n_1$. Note that, any changes to a user's stake can only be applied and updated after a certain check points. For simplicity, we assume that stake synchronization is done between nodes for every 20th frame. 

\subsection{Node Structure}
This section gives an overview of the node structure in \stakedag. A node in \stakedag\ system is associated with a stakeholder that participates with other nodes in the network.

\begin{figure}[ht]
	\centering
	\includegraphics[width=0.6\linewidth]{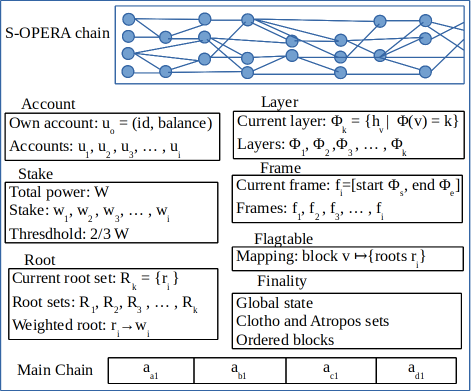}
	\caption{\stakedag\ node structure}
	\label{fig:nodestructure}
\end{figure}

Figure~\ref{fig:nodestructure} shows the structure of a node. 
Each node consists of an S-OPERA (Stake-based) chain of event blocks, which is the local view of the DAG.
Each node stores the account information of all nodes including itself. The account information includes the account id and its current token helding in each account. There are other details including tokens, delegated tokens, as described in Section~\ref{se:staking}. For the sake of simplicity, we only show the account balance (e.g., number of tokens).
The stakes are then computed from account balance of all nodes. Total validating power and validating threshold are updated accordingly.

In each node, it stores the layering information (the resulting H-OPERA chain), roots, root graph, frames, Clotho set and Atropos set.
The top event block of each peer is the most recently created / received event block by that peer node.
A block can become a root if it receives more than the validating threshold, which is 2/3 of the validating power of the network. A root will become a Clotho when it is further known by another root. 
A frame contains a set of roots of that frame. 
When a root becomes Clotho, its Clotho flag is set on and then it is sorted and assigned a final ordering number. After getting the final ordering, the Clotho is then promoted to an Atropos. The Main-chain is a data structure storing hash values of the Atropos blocks. 

\subsection{Event block creation}
Like previous Lachesis protocol~\cite{lachesis01}, $S_\phi$ protocol allows every node to create event blocks. In order to generate a new event block, a node $i$ sends synchronization requests to $k$ - 1 other nodes to get their latest event blocks. Then node $i$ can create a new event block. A new event block has $k$ references: one is self-ref referencing the top event block of the same node, $k$-1 other references referencing the top event blocks of $k$-1 peers. With cryptographic hashing, an event block can only be created or added if the self-ref and other-ref event blocks exist.

\subsection{Layering}\label{se:consen}
We then present the main concepts and algorithms of the $S_{\phi}$ protocol in our \stakedag\ framework. This layering-based PoS approach is based on our ONLAY paper~\cite{onlay19}. 
Intuitively, $S_{\phi}$ integrates layering algorithms on the S-OPERA chain and then use the assigned layers reach consensus of the event blocks.

For a DAG $G$=($V$,$E$), a layering $\phi$ of $G$ is a mapping $\phi: V \rightarrow Z$, such that for any directed edge ($u$,$v$) $\in E$, the layer of the source $v$ is greater than the layer of the sink $u$; that is, $\phi(v)$ $\geq$ $\phi(u)$ + 1. 
Thus, $\phi$ partitions the set of vertices $V$ into a finite number of \emph{non-empty} disjoint subsets (called layers) $V_1$,$V_2$,$\dots$, $V_l$, such that $V$ = $\cup_{i=1}^{l}{V_i}$. Each vertex is assigned to a layer $V_i$.

Recall that S-OPERA chain is a DAG $G$=($V$,$E$) stored in each node. Each vertex in $V$ is an event block, which has a validation score. By applying $\phi$ on the S-OPERA chain, one can obtain the hierarchical graph of $G$, which is called H-OPERA chain, which is the resulting hierarchical graph $H = (V,E, \phi)$. One can apply either LongestPathLayer (LPL) algorithm or Coffman-Graham (CG) algorithm on the graph $G$ (see details in ~\cite{onlay19}).

In \stakedag, the S-OPERA chain is evolunary as each node creates and synchronizes events with each other.Let us consider a model of dynamic S-OPERA chain. For a node $i$, let $G$=($V$,$E$) be the current S-OPERA chain and $G'$=($V'$,$E'$)
denote the \emph{diff graph}, which consists of the changes to $G$ at a time, either at block creation or block arrival.The vertex sets $V$ and $V'$ are disjoint; similar to the edge sets $E$ and $E'$. At each graph update, the updated S-OPERA chain becomes $G_{new}$=($V \cup V'$, $E \cup E'$). 

%\subsubsection{Dynamic S-OPERA chain}

\subsubsection{Online Layer Assignment}
We adopt the online layering algorithms introduced in~\cite{onlay19}. 
Specifically, we use \emph{Online Longest Path Layering} (O-LPL) and \emph{Online Coffman-Graham} (O-CG) algorithm, which assign layers to the vertices in diff graph $G'$=($V'$,$E'$) consisting of new self events and received unknown events.
The algorithms are efficient and scalable to compute the layering of large dynamic S-OPERA chain.
Event blocks from a new graph update are sorted in topological order before being assigned layering information. The two algorithms are presented in Appendix~\ref{se:app-onlinelayering}.

\subsubsection{Layer Width Analysis in BFT}\label{sec:layerwidth}
We then give ourformal analysis of the layering algorithms with respect to the Byzantine fault tolerance.
BFT addresses the functioning stability of the network when Byzantine nodes may take up to $\mathcal{W}/3$ validating power, where $\mathcal{W}$ is the total validating power of the network.

Let $\phi_{LP}$ denote the layering of $G$ obtained from LongestPathLayering algorithm. Let $\phi_{CG(W_{max})}$ denote the layering of $G$ obtained from the Coffman-Graham algorithm with some fixed width $W_{max}$. 

{\bf Byzantine free network} We first consider the case in which the network is free from any faulty nodes; all nodes are honest and thus zero fork exists.

\begin{prop}
	In the absense of dishonest nodes, the width of each layer $L_i$ is at most $n$.
\end{prop}

This can be easily proved by induction. Since there are $n$ nodes, there are $n$ leaf vertices. The width of the first layer $L_1$ is $n$ at maximum, otherwise there exists a fork, which is not possible. We assume that the theorem holds for every layer from 1 to $i$. That is, the $L_i$ at each layer has the width of at most $n$. We will prove it holds for layer $i+1$ as well.
Since each honest node can create at most one event block on each layer, the width at each layer is at most $n$.
We can prove by contradiction. Suppose there exists a layer $|\phi_{i+1}| > n$. Then there must exist at least two events $v_p$, $v_q$ on layer $\phi_{i+1}$ such that $v_p$ and $v_q$ have the same creator, say node $n_i$. That means $v_p$ and $v_q$ are the forks of node $n_i$. It is a contradiction with the assumption that there is no fork. Thus, the width of each layer $|\phi_{i+1}| \leq n$. It is proved.

Thus, in the absence of forks, LongestPathLayering(G) and CoffmanGraham(G,n) computes the same layer assignment. That is, $\phi_{LP}$ and $\phi_{CG(n)}$ are identical; or
 $\forall v \in G$, $\phi_{LP}(v)$ = $\phi_{CG(n)}(v)$.

{\bf 1/3-BFT} Now, let us consider the layering result for the network with at most $\mathcal{W}/3$ validating power are compromised. In this case, there are about $n/3$ dishonest nodes in average, and about $n$-1 dishonest nodes in the worst case.

Without loss of generality, let $w_p$ be the probability that a node can create fork.
Let $w_c$ be the maximum number of forks a faulty node can create at a time. The number of forked events a node can create at a time is $w_p w_c$.
The number of fork events at each layer is given by 
$W_{fork} = (n-1) w_p w_c$.
Thus, the maximum width of each layer is given by:
$$W_{max} = n + W_{fork} = n + (n-1) w_p w_c \approx  (1+w_p w_c)n$$

Thus, if we set the maximum width high enough to tolerate the potential existence of forks, we can achieve BFT of the layering result. The following theorem states the BFT of LPL and CG algorithms.

LongestPathLayering($G$) and CoffmanGraham($G$, $W_{max}$) computes the same layer assignment. That is, $\phi_{LP}$ and $\phi_{CG(W_{max})}$ are identical; or for each vertex $v \in G$, $\phi_{LP}(v) = \phi_{CG(W_{max})}(v)$.

\subsubsection{Root Graph}\label{se:rootgraph}
With layering information, we propose to a new data structure, called \emph{root graph}, 
which gives a more efficient way to compute new roots. The root graph can be used together with the layers for fast computation of frames.
This data structure is slightly different from the one in~\cite{onlay19} as we use validating powers of roots instead of the root count.

\dfnn{Root graph}{
	A root graph $G_R$=($V_R$, $E_R$) is a directed graph consisting of roots as vertices, and  their reachable connections as edges.}

The root graph $G_R = (V_R, E_R)$ contains vertices as roots $V_R  \subseteq V$, and the set of edges $E_R$  the reduced edges from $E$, such that $(u,v) \in E_R$ only if $u$ and $v$ are roots and there is a path from $u$ to $v$ following edges in $E$.
The root graph initially contains the $n$ genesis vertices --- leaf event blocks. When a vertex $v$ reaches $2\mathcal{W}/3$ of the validating power of the current root set $V_R$, it becomes a root. For each root $r_i$ that new root $r$ reaches, we include a new edge ($r$, $r_i$) into the set of root edges $E_R$. If a root $r$ reaches two roots $r_1$ and $r_2$ of the same node, then we retain only one edge ($r$, $r_1$) or ($r$, $r_2$) if $\phi(r_1) >\phi(r_2)$. This requirement makes sure each root of a node can have at most one edge to any other node.

\begin{algorithm}
	\caption{Root graph algorithm}\label{al:buildrootgraph}
	\begin{algorithmic}[1]
		\State Require: H-OPERA chain $H$
		\State Output: root graph $G_R=(V_R,E_R)$
		\State{$R \leftarrow$ set of leaf events}
		\State $V_R \leftarrow R$
		\State $E_R \leftarrow \emptyset$
		
		\Function{buildRootGraph}{$H$, $\phi$, $l$}
		\For{each layer $i$=1..$l$}
		\State$Z \leftarrow \emptyset$
		\For{each vertex $v$ in layer $\phi_i$}
		\State $S \leftarrow$ the set of vertices in $R$ that $v$ reaches
		\State $w_S$ $\leftarrow$ the total weights of roots in $S$ (*)
		\If{$w_S > 2W/3$} (*)
		\For{each root $r_i  \in S$}
		\State $E_R \leftarrow E_R \cup \{(v,r_i)\}$
		\EndFor
		\State $V_R \leftarrow V_R \cup \{v\}$
		\State $Z \leftarrow Z \cup \{v\}$ 
		\EndIf
		\EndFor
		
		\For{each vertex $v_j \in Z$}
		\State Let $v_{old}$ be a root in $R$ such that $cr(v_j) = cr(v_{old})$		
		\State $R \leftarrow R \setminus \{v_{old}\} \cup \{v_j\}$
		\EndFor
		\EndFor
		\EndFunction
	\end{algorithmic}
\end{algorithm}

\subsection{Root selection}
For \stakedag, we propose a new approach that uses root graph and frame assignment.
One can also use a Root selection algorithm described in our previous paper~\cite{fantom18}.

The steps to build the root graph out of an H-OPERA chain is given in Algorithm~\ref{al:buildrootgraph}. 
At the start of each layer, $Z$ is set to an empty set.
The algorithm processes from layer 0 to the maximum layer $l$. $R$ is the set of current active roots. At initial stage, $R$ contains only the  the $n$ genesis vertices --- leaf event blocks (Line 3). 
We will update the set after finding new roots at each layer. Line 4 and 5 set the initial value of the vertex set and edge set of the root graph. For each layer $i$, it runs the steps from line 8 to 18. Let $Z$ denote the set of new roots found at layer $i$. $Z$ is initially set to empty (Line 8). We then process each vertex $v$ in the layer $\phi_i$ (line 9). We compute the set $S$ of all active roots in $R$ that $v$ can reach (Line 10). The total weight $w_S$ is computed as the total weights of all roots in $R$ (Line 11). $w_S$ is the validation score of $v$. If $w_S$ is greater than $2\mathcal{W}/3$, $v$ becomes a new root (line 12). For every root $r_i$ in $S$, we add an edge $(v,r_i)$ into $E_R$ (line 13-14).
In lines 15-16, we add $v$ into $V_R$, and also into $Z$ to keep to keep track of the new roots found at layer $i$. 
The second inner loop will replace any new roots found at layer $i$ with the existing ones in current active set $R$ (line 17-19).

The algorithm always updates the active set $R$ of roots at each layer. The number of active roots in $R$ is always $n$, though the active roots can be from different layers.

\subsection{Clotho selection}

For pBFT, consensus is reached once an event block is known by more than 2/3$n$ of the nodes and that information is further known by 2/3$n$ of the nodes.

In our PoS $S_\phi$ protocol, we define the consensus of a root as the condition the root becomes a Clotho. If it can be reached by more than $2\mathcal{W}/3$ validating power of the roots and the information is confirmed by another $2\mathcal{W}/3$ validating power of the roots.
That is, for a root $r$ at frame $f_i$, if there exists a root $r'$ at a frame $f_j$ such that $r'$ reaches $r$ and $j$ $\geq$ $i$+2, then the root $r$ reaches pBFT consensus. The root $r'$ is called the \emph{nominator}, which nominates $r$ to become a Clotho.

\begin{prop}[Global Clotho]
	For any two honest nodes $n_i$ and $n_j$, if $c$ is a Clotho in $H_i$ of $n_i$, then $c$ is also a Clotho in $H_j$ of $n_j$.
\end{prop}

Assume a block $r$ that exists in both nodes $n_i$ and $n_j$. Since honest nodes are consistent nodes, then the subgraphs $G[r]$ and $G'[r]$ are the same on both nodes.
Suppose $r$ nominates a root $c$ to be a Clotho in $n_i$. Because the subgraphs are identical, it can be deduced that $r$ also nominates the same root $c$ as a Clotho in $n_j$.

\subsection{Atropos Selection}

When new Clotho event blocks are found, we will order them and then will assign the final \emph{consensus time} for them. We first sort the Clothos based on topological ordering. Section~\ref{se:toposort} gives more details on the topological ordering and consensus time assignment. 

Once a Clotho is assigned with a consensus time, it becomes an Atropos.
Each node stores the hash value of Atropos and Atropos consensus time in the Main-Chain (blockchain). The Main-chain is used for quick retrieval of the final ordering between event blocks.

\begin{prop}[Global Atropos]
	For any two honest nodes $n_i$ and $n_j$, if $a$ is an Atropos in $H_i$ of $n_i$, then $a$ is also an Atropos in $H_j$ of $n_j$.
\end{prop}

From the above Global Clotho proposition, a root is a Clotho in $n_i$ then it is also a Clotho in $n_j$. The layering is deterministic and so is the topological sorting. Hence, the sorted list of Clothos on two nodes are identical. Thus, the assigned consensus time for the Clotho $c$ is the same on both nodes. An Atropos block $a$ in $n_i$ is also an Atropos in $n_j$, and they have the same consensus time.

\subsection{Frame Assignment}\label{se:frameassignment}

We then present a deterministic approach to frame assignment, which assigns each event block a unique frame number.
First, we show how to assign frames to root vertices via the so-called \emph{root-layering}. 
For a root graph $G_R=(V_R, E_R)$, we assign frame number to roots $v_r$ of $V_R$ as a function $\phi_R$, as follows:
\begin{itemize}
	\item $\phi_R(u) \geq \phi_R(v) + 1$,  $\forall (u,v) \in E_R$.
	\item if $u$ reaches at least $2\mathcal{W}/3$ validating power of the roots of frame $i$, then $\phi_R(v)$ = $i + 1$.
\end{itemize}

Second, we then assign frame numbers to non-root vertices with respect to the topological ordering of the vertices. The frame assignment to roots are used to assign frames to non-root vertices. 
For vertices in the layers between layer $i$ and layer $i$ +1, the vertices are assigned the frame number $i$.

\subsection{Fork detection and removal}

When a node $n$ receives a Sync request from another node, $n$ will check if the received events would cause a fork. In 1/3-BFT system, we can prove that there exists an honest node that sees the fork. That honest node will remove one of the forked events from its S-OPERA chain and will then notify the fork to all other nodes.

\dfnn{Fork}{
	Two events $v_x$ and $v_y$ form a \emph{fork} if they have the same creator, but neither is a self-ancestor of the other. Denoted by $v_x \efork v_y$.}

The fork relation is symmetric; that is $v_x \efork v_y$ iff $v_y \efork v_x$.
For two events $v_x$ and $v_y$ of the same creator, we can prove the following: $v_x \efork v_y$ iff $v_x \concur v_y$. By definition of fork, ($v_x$, $v_y$) is a fork implies that $v_x \not \eancestor v_y$ and $v_y \not \eancestor v_x$. It means $v_x \not \rightarrow v_y$ and $v_y \not \rightarrow v_x$ (by Happened-Before) and thus $v_x \concur v_y$ (by definition of concurrent).

\begin{lem} (Fork Detection). If there is a fork $v_x \efork  v_y$, then $v_x$ and $v_y$ cannot both be roots on two honest nodes.
\end{lem}
\begin{proof}
	Since any honest node cannot accept a fork, $v_x$ and $v_y$ cannot be roots on a single honest node. Now we prove a more general case, showing a proof by contradiction. 
	
	Suppose that $v_x$ is a root of $n_x$, and $v_y$ is root of $n_y$, where $n_x$ and $n_y$ are two distict honest nodes. Because $v_x$ is a root, it reached roots of node set $S_1$ whose weights account for more than $2/3$ of total validating power. Similarly, $v_y$ is a root, it reached roots of node set $S_2$ that has a total weight greater than $2\mathcal{W}/3$. Thus, there must exist an overlap of two sets $S_1$ and $S_2$, whose total weights is more than $W/3$.
	For 1/3-BFT system, it is assumed that dishonest nodes account for less than $W$/3 of the power. Hence, there must be at least one honest member in the overlap set. Let $n_h$ be such an honest member. As $n_h$ is honest, $n_h$ does not accept the fork. This gives a contradiction. The lemma is proved.
\end{proof}

When an honest node finds a new root $r$, there is a guarantee that there exists no forked events  in the subgraph under $r$. If there were, the fork should have been detected and removed prior to the promotion of $r$.

\begin{prop}[Fork-free]
	For any Clotho $c$, the subgraph $G[c]$ is fork-free.
\end{prop}

When $c$ became a root, any forks under the subgraph $G[c]$ has already been detected and removed from $G$. The node then notified / synchronized with peers about the forks it found, and also might receive any notified forks from other peers. By definition of Clotho, $c$ is known by roots reached $2\mathcal{W}/3$ validating power and those are in turn known by another set of roots of $2\mathcal{W}/3$ validating power. Thus, there exist nodes with more than $2\mathcal{W}/3$ validating power knew about and removed the detected forks via peer synchronization.

\begin{prop}[Fork-free Global chain]
	The global consistent OPERA chain $G^C$ is fork-free.
\end{prop}

Recall that the global consistent chain $G^C$ consists of the finalised events i.e., the Clothos and its subgraphs that get final ordering at consensus. Those finally ordered Clothos are Atropos vertices. From the above proposition, for every Clotho $v_c$, there is no fork in its subgraph  $G[v_c]$. Thus, we can prove by processing the Clothos in topological order from lowest to highest.

\subsection{Topological sort}\label{se:toposort}

After new Clothos are found, the protocol will compute the final ordering for the Clothos as well as the vertices under them.
Once the final total order is computed, finalized vertices are assigned with a consensus time.

We present an approach to ordering event blocks that reach finality.
Algorithm~\ref{algo:topoordering} gives our topological sort algorithm.
The algorithm takes as input the set of new Clothos $C$, the
H-OPERA chain, layering assignment $\phi$, frame assignment $\phi_S$, current ordered list $S$, and set of all ordered blocks $U$. 
Initially, the list $S$ and the set $U$ are empty (line 1-2).

We first order the Clothos vertices using $SortByLayer$ function, which sorts vertices based on their layer, then lamport timestamp and then hash information of the event blocks (line 12-13). Second, the algorithm processes every Clotho in the sorted order (line 5). For each Clotho $c$, the algorithm computes the subgraph $G[c] = (V_c, E_c)$ under $c$ (line 6). The set of vertices $V_u$ contains vertices from $V_c$ that is not yet processed in $U$ (line 7).
We then apply $SortByLayer$ to order vertices in $V_u$ to get an ordered list of events $T$. For each event block $v$ in $T$, it appends $v$ into the final ordered list $S$ and also adds $v$ into the set of already processed vertices $U$.

\begin{algorithm}[H]
	\caption{Topological Sort}\label{algo:topoordering}
	\begin{algorithmic}[1]
		\State $S \leftarrow$ empty list
		\State $U \leftarrow \emptyset$
		
		\Function{TopoSort}{$C$, $H$, $\phi$, $\phi_F$, $S$, $U$}
		\State $Q \leftarrow$ $SortByLayer(C)$
		\For{each Clotho $c \in Q$}
		\State Compute the graph $G[c] = (V_c, E_c)$
		\State $V_u \leftarrow V_c \setminus U$
		\State $T \leftarrow SortByLayer(V_u)$
			\For{$v \in T$}
			\State Append $v$ at the end $S$
			\State $U \leftarrow U \cup \{v \}$
			\EndFor
		\EndFor
		\EndFunction
		\Function{SortByLayer}{$K$}
		\State Sort the vertices in $K$ by layer, Lamport timestamp and hash in that order.
		\EndFunction
	\end{algorithmic}
\end{algorithm}

After the final order of event blocks are computed using the above algorithm, we can assign the consensus time to the finally ordered event blocks.
We can prove that the algorithm will sort and assign the final ordering index for each event block once. A proof by induction is fairly straight-forward.

\subsection{Transaction confirmations}

Here are some steps for a transaction to reach finality in our system, similar to our previous paper~\cite{onlay19}.
In a successful scenario, five confirmations will be issued and the fifth receipt is the final confirmation of a successful transaction.
First, after a client submits a transaction, s/he will be issued a confirmation receipt of the submitted transaction. 
Second, the node will batch the submitted transaction(s) into a new event block added into the node's DAG and it will broadcast the event block to all other nodes of the system. A receipt will be then issued to confirm the containing event block identifier is being processed. 
Third, when the event block receives the majority of the validating power (e.g., it becomes a Root block), or being known by such a Root block, a confirmation will be issued to acknowledge the block has reached that 2/3 validating power. 
Fourth, when the Root event block becomes a Clotho. A confirmation will be sent to the client showing that the event block has come to the semi-final stage as a Clotho or being confirmed by a Clotho.
Fifth, after the Clotho stage, we will determine the consensus timestamp for the Clotho and its dependent event blocks. Once an event block gets the final consensus timestamp, it is finalized and a final confirmation will be issued to the client that the transaction has been successfully finalized. 

In some cases, a submitted transaction can fail to reach finality. For example, a transaction does not pass the validation due to insufficient account balance, or violation of account rules. The other kind of failure is when the integrity of DAG structure and event blocks is not complied due to the existence of compromised or faulty nodes. In such unsuccessful cases, the event blocks are marked for removal and detected issues are notified to all nodes. Receipts of the failure will be sent to the client.

%% file: discuss.tex
\newpage
\section{Discussions}\label{se:discuss}

This section presents several important discussions about PoW and PoS previous work. There has been extensive research in the protocol fairness and security aspects of existing PoW and PoS blockchains (see surveys~\cite{sheikh2018proof, panarello2018survey}). Then we also highlight the benefits of our proposed DAG-based PoS approach.

\subsection{Protocol Fairness}

Here, the common concerns about fairness in PoW and PoS in previous protocols are given as follows:
\begin{itemize}
	\item PoW protocol is fair in the sense that a miner with $p$ fraction of the total computational power can win the reward and create a block with the probability $p$.
	\item PoS protocol is fair given that an individual node who has $p$ fraction of the total number of coins in circulation creates a new block with $p$ probability. In a PoS system, there is always a concern that the initial holders of coins will not have an incentive to release their coins to third parties, as the coin balance directly contributes to their wealth.
\end{itemize}

Our \stakedag\ protocol is fair because every node has an equal chance to create an event block. A protocol in $\mathfrak{S}$ family allows node to enter the network without the need to equip an expensive and high-spec hardward like in PoW. Further, any node in $\mathfrak{S}$ protocol can create a new event block with the same propability, unlike stake-based probability of block creation in POS blockchains.

Like a PoS blockchain system, it is a possible concern that the initial holders of coins will not have an incentive to release their coins to third parties, as the coin balance directly contributes to their wealth. Unlike PoS, that concern in \stakedag\ protocol is about the economic rewards a \stakedag\ node may get is proportional to the stake they possess after they successfully contribute to the validation of event blocks.

Remarkably, our \stakedag\ protocol is more intuitive because our reward model used in stake-based validation can lead to a more reliable and sustainable network.

\subsection{Validation Procedure}
For validation in \stakedag, all the validators will be weighted by the tokens in their deposit. A block reaches consensus when it gains more than 2/3 of the network stake. That means, it is agreed by a set of validators whose stakes are greater than 2/3 of the total network stake.

Our DAG-based PoS approach makes some improvements in validation procedure over the previous PoS approaches. In particular, \stakedag\ protocol utilizes the following validation procedure, which is based on the Casper PoS model~\cite{buterin2018}. Our approach can be summarised as follows:
\begin{itemize}
	\item  Each block is 1-20 seconds and each frame is 1-10 minutes. Every 20th frame is a checkpoint. Stakeholders can choose to make more deposits at each check point, if they want to become validators and earn more rewards. Validators can choose to exit, but cannot withdraw their deposits until three months later.
	\item With asynchronous system model, frames and accounts may be not synchronized and validators are incentivized to coordinate on which checkpoints the history should be updated. This coordination is carried out by broadcasting the latest local views amongst nodes.
	\item A checkpoint is selected based on a consistent global history that is finalized with more than 2/3 of the validating power for the checkpoint. When a checkpoint is finalized, the transactions will not be reverted. Attackers can attempt double voting to target double spending attachs. Honest validators are incentivized to report such behaviors and burn the deposits of the attackers.
\end{itemize}

\subsection{Security}

PoS approach reduces the energy demand, compared to PoW. PoS is considered as more secure than PoW.
This section gives some security analysis of PoW, PoS and our \stakedag\ protocol.

{\bf Overall Comparison} As for a guidelines, we present a summary that gives the effects of the common types of attack  on previous protocols.

\begin{table}[h]
	\centering
	\begin{tabular}{|l|c|c|c|c|}
		\hline
		Attack type & PoW & PoS & DPoS & \stakedag \\
		\hline
		Short range attack (e.g., bribe) & - & + & - & - \\
		Long range attack  & - & + & + & maybe \\
		Coin age accummulation attack & - & maybe & - & maybe \\
		Precomputing attack & - & + & - & - \\
		Denial of service & + & + & + & + \\
		Sybil attack & + & + & + & maybe \\
		Selfish mining & maybe & - & - & - \\
		\hline
	\end{tabular}
	\vspace{5pt}
	\caption{Comparison of Vulnerability Between PoW, PoS, DPoS and \stakedag\ Protols}\label{tab:vulnerability}
\end{table}
Table~\ref{tab:vulnerability} gives a comparison between PoW, PoS, DPoS and our \stakedag\ protocol. Generally speaking, \stakedag\ has less vulnerabilities than PoW, PoS and DPoS.

{\bf PoW vulnerabilities}
PoW-based systems are facing \emph{selfish mining attack}. In selfish mining, an attacker selectively reveals mined blocks in an attempt to waste computational resources of honest miners.

{\bf PoS vulnerabilities}
PoS has encountered new issues arise that were not present in PoW-based blockchains. These issues are: (1) \emph{Grinding attack:} malicious nodes can play their bias in the election process to gain more rewards or to double spend their money: (2) \emph{Nothing at stake attack:} A malicious node can mine on an alternative chain in PoS at no cost, whereas it would lose CPU time if working on an alternative chain in PoW.

{\bf Shared vulnerabilities in PoW and PoS} There are several vulnerabilities that are encountered by both PoW and PoS. \emph{DoS attack} and \emph{Sybil attack} are shared 
vulnerables that PoW are found more vulnerable. A DoS attack disrupts the network by flooding the nodes. In a Sybil attack, the attacker creates numerous faulty nodes to disrupt the network.

For another shared vulnerable \emph{Bribe attack}, PoS is more vulnerable because a PoS Bribe attack costs 50x lower than PoW Bribe attack. In bribing, the attacker performs a spending transaction, and at the same time builds an alternative chain secretely, based on the block prior to the one containing the transaction. After the transaction gains the necessary number of confirmations, the attacker publishes his chain as the new valid blockchain, and the transaction is reversed.

{\bf PoS potential attacks}
We then discuss two scenarios that are possible attacks in PoS. Both of these attacks can induce conflicting finalized checkpoints that will require offline coordination by honest users.

\emph{Double-Spending}
An attacker (a) acquires $2F/3$ of stakes; (b) submits a transaction to spend some amount and then votes to finalize a checkpoint (Atropos) that includes the transactions; (c) sends another transaction to double-spends; (d) gets caught and his stakes are burned as honest validators are incentivezed to report such misbehavior. 
In another scenario, an attacker acquires (a) $F/3+\epsilon$ to attempt for an attack and suppose Blue validators own $F/3-\epsilon/2$, and Red validators own the rest $F/3-\epsilon/2$. The attacker can (b) vote for the transaction with Blue validators, and then(c) vote for the conflicting transaction with the Red validations. Then both transactions will be finalized because they have $2F/3 +\epsilon/2$ votes. Blue and Red validators may later see the finalized checkpoint, approve the transaction, but only one of them will get paid eventually.

\emph{Sabotage (going offline)}  An attacker owning $F/3 +\epsilon$ of the stakes can appear offline by not voting and hence checkpoints and transactions cannot be finalized. Users are expected to coordinate outside of the network to censor the malicious validators.

In \stakedag, the protocol relies on every honest validators to detect and report such attacks. The attack, once detected, will cause a loss of tokens to the attacker. In fact, our \stakedag\ protocol guarantees a 1/3-BFT in which no more than $F/3$ tokens are owned by misbehaving nodes.
We have provided our proof of 1/3-BFT in previous section and also more details are given in the Appendix. 

Specially, let us consider the case an attacker in \stakedag\ network with $F/3 + \epsilon$. In Double-Spending, s/he can manage to achieve a root block $r_b$ with Blue validators and a conflict root block $r_r$ with Red validators. However, in order for either of the two event blocks to be finalized (becoming Clotho and then Atropos), each of root blocks need to be confirmed by two roots of next levels. Since peers always share their event blocks, an attacker cannot stop the Blue and Red validators (honest) to share event blocks to each other. Therefore, there will exist an honest validator from Blue or Red groups, who detects conflicting event blocks in the next level roots. Because honest validators are incentivezed to find out and report such wrongdoing, the attacker will get caught and his stakes are burned. Similarly, for Sabotage attack, the attacker may refuse to vote for a while. But honest validators will find out the absence of those high stake validators after a number of computed layers.

{\bf Comparison of attack cost in PoS versus PoW} It will cost more for an attack in PoS blockchain due to the scarity of the native token than in a PoW blockchain.
In order to gain more stake for an attack in PoS, it will cost a lot for an outside attacker. 
S/he will need $2F/3$ (or $F/3$ for certain attacks) tokens, where $F$ is the total number of tokens, regardless of the token price. Acquiring more tokens will definitely increase its price, leading to a massive cost. Another challenge is that all the tokens of a detected attempt will be burned.
In contrast, PoW has no mechanism nor enforcement to prevent an attacker from reattempting another attack. An attacker may purchase or rent the hash power again for his next attempts.

Like PoS, our \stakedag\ protocol employs  the PoS mechanism to effectively prevent potential attacks. Attackers will need to acquire $2F/3$ tokens (or at least $F/3$ for certain attacks) to fully influence the validation process. Any attempt that is detected by peers will void the attacker's deposit.

%% file: appendix.tex
\newpage
\section{Appendix}\label{se:appendix}

This section gives further details about the \stakedag\ protocol. We present the formal semantics of  $S_\phi$ using the concurrent common knowledge that can be applied to a generic model of DAG-based PoS approaches, and then show the proofs for consensus of the protocol.

\subsection{Formal definitions}\label{se:prelim}

\subsubsection{Node State}
A node is a machine participating in the $\mathfrak{S}$ protocol.
Each node has a local state consisting of local histories, messages, event blocks, and peer information.  Let $n_i$ denote the node with the identifier of $i$. Let $n$ denote the total number of nodes.

\dfnn{State}{A (local) state of $i$ is denoted by $s_j^i$ consisting of a sequence of event blocks $s_j^i$=$v_0^i$, $v_1^i$, $\dots$, $v_j^i$.}

In a DAG-based protocol, each  event block $v_j^i$ is \emph{valid} only if the reference blocks exist before it. A local state $s_j^i$ is corresponding to a unique DAG. In \stakedag, we simply denote the $j$-th local state of a node $i$ by the DAG $g_j^i$. Let $G_i$ denote the current local DAG of a process $i$.

An action is a function from one local state to another local state. An action can be either: a $send(m)$ action of a message $m$, a $receive(m)$ action, and an internal action. A message $m$ is a triple $\langle i,j,B \rangle$ where the sender $i$, the message recipient $j$, and the message body $B$. 
In \stakedag, $B$ consists of the content of an event block $v$. Let $M$ denote the set of messages.
Semantics-wise, there are two actions that can change a process's local state: creating a new event and receiving an event from another process.\\

\dfnn{Event}{An event is a tuple $\langle  s,\alpha,s' \rangle$ consisting of a state, an action, and a state.}

Sometimes, the event can be represented by the end state $s'$. 
The $j$-th event in history $h_i$ of process $i$ is $\langle  s_{j-1}^i,\alpha,s_j^i \rangle$, denoted by $v_j^i$.

\dfnn{Local history}{A local history $h_i$ of $i$ is a sequence of local states starting with an initial state. A set $H_i$ of possible local histories for each process $i$.}

A process's state can be obtained from its initial state and the sequence of actions or events that have occurred up to the current state. \stakedag\ protocol uses append-only semantics. The local history may be equivalently described as either of the following: (1)
$h_i$ = $s_0^i$,$\alpha_1^i$,$\alpha_2^i$, $\alpha_3^i$ $\dots$, (2)
$h_i$ = $s_0^i$, $v_1^i$,$v_2^i$, $v_3^i$ $\dots$, (3)
$h_i$ = $s_0^i$, $s_1^i$, $s_2^i$, $s_3^i$, $\dots$.
In \stakedag, a local history is equivalently expressed as:
$h_i$ = $g_0^i$, $g_1^i$, $g_2^i$, $g_3^i$, $\dots$
where $g_j^i$ is the $j$-th local DAG (local state) of the process $i$.

\dfnn{Run}{Each asynchronous run is a vector of local histories. Denoted by
	$\sigma$ = $\langle h_1,h_2,h_3,...h_N \rangle$.}

Let $\Sigma$ denote the set of asynchronous runs. A global state of run $\sigma$ is an $n$-vector of prefixes of local histories of $\sigma$, one prefix per process. The happens-before relation can be used to define a consistent global state, often termed a consistent cut, as follows.

\subsubsection{Lamport timestamps and Ordering}

Our \stakedag\ protocol relies on Lamport timestamps to define a topological ordering of event blocks.
The ``happened before" relation, denoted by $\rightarrow$, gives a partial ordering of event blocks in a distributed system.
For a pair of $v_i$ and $v_j$, then $v_i$ $\rightarrow$ $v_j$ if : (1) $v_i$ and $v_j$ are events of the same node $p_i$, and $v_i$ comes before $v_j$, (2) $v_i$ is the send($m$) by one process and $v_j$ is the receive($m$) by another process (3) $v_i$ $\rightarrow$ $v_k$ and $v_k$ $\rightarrow$ $v_j$ for some $v_k$. 

\dfnn{Happened-Im-Before}{An event block $v_x$ is said Happened-Immediate-Before an event block $v_y$ if $v_x$ is a (self-) ref of $v_y$. Denoted by $v_x$ $\hibefore$ $v_y$.}

\dfnn{Happened-before}{An event block $v_x$ is said Happened-Before an event block $v_y$ if $v_x$ is a (self-) ancestor of $v_y$. Denoted by $v_x$ $\hbefore$ $v_y$.}

Happened-before relation is the transitive closure of happens-immediately-before.
Two event blocks $v_x$ and $v_y$ are said \emph{concurrent}, denoted by $v_x \concur v_y$, if neither of them  happened before the other. Two distinct events $v$ and $v'$ are said to be concurrent if $v \nrightarrow v'$ and $v' \nrightarrow v$.
Given two vertices $v_x$ and $v_y$ both contained in two OPERA chains (DAGs) $G_1$ and $G_2$ on two nodes. We have the following:
(1) $v_x$ $\hbefore$ $v_y$ in $G_1$ if $v_x$ $\hbefore$ $v_y$ in $G_2$;  (2) $v_x$ $\concur$ $v_y$ in $G_1$ if $v_x$ $\concur$ $v_y$ in $G_2$.

\dfnn{Total ordering}{
	Let $\prec$ denote an arbitrary total ordering  of the nodes (processes) $p_i$ and $p_j$. 	\emph{Total ordering} is a relation $\Rightarrow$ satisfying the following: for any event $v_i$ in $p_i$ and any event $v_j$ in $p_j$, $v_i \Rightarrow v_j$ if and only if either (i) $C_i(v_i) < C_j(v_j)$ or (ii) $C_i(v_i)$=$C_j(v_j)$ and $p_i \prec p_j$.}

This defines a total ordering relation. The Clock Condition implies that if $v_i \rightarrow v_j$ then $v_i \Rightarrow v_j$. 

\subsubsection{Consistent Cut}\label{sec:cck} 

An asynchronous system consists of the following sets: a set $P$ of process identifiers, a set $C$ of channels, a set $H_i$ of possible local histories for each process $i$, a set $A$ of asynchronous runs, a set $M$ of all messages.
Consistent cuts represent the concept of scalar time in distributed computation, it is possible to distinguish between a ``before'' and an ``after'', see CCK paper~\cite{cck92}. 

Consistent cut model of \stakedag\ protocol is based on the model in our previous work~\cite{fantom18,onlay19}. More details can be found in our previous papers.

\subsection{DAG and S-OPERA chain}

Each node can create event blocks, send (receive) messages to (from) other nodes. 
Each event block has exactly $k$ references: a self-ref reference, and $k$-1 other-ref references pointing to the top events of $k$-1 peer nodes.
Each node stores the current local state in a form of a directed acyclic graph $G$=($V$, $E$), where $V$ is a set of vertices and $E$ is a set of edges.

\dfnn{Event block}{An event block is a holder of a set of transactions. An event block includes the signature, generation time, transaction history, and references to previous event blocks.}

\dfnn{Top event}{An event $v$ is a top event of a node $n_i$ if there is no other event in $n_i$ referencing $v$.}

\dfnn{Ref}{An event $v_r$ is called ``ref" of event $v_c$ if the reference hash of $v_c$ points to the event $v_r$. Denoted by $v_c \eref v_r$. For simplicity, we can use $\erefz$ to denote a reference relationship (either $\eref$ or $\eself$).}

\dfnn{Self-ref}{An event $v_s$ is called ``self-ref" of event $v_c$, if the self-ref hash of $v_c$ points to the event $v_s$. Denoted by $v_c \eself v_s$.}

An event block $v_a$ is \emph{self-ancestor} of an event block $v_c$ if there is a sequence of events such that $v_c \eself v_1 \eself \dots \eself v_m \eself v_a $. Denoted by $v_c \eselfancestor v_a$.
An event block $v_a$ is an \emph{ancestor} of an event block $v_c$ if there is a sequence of events such that $v_c \erefz v_1 \erefz \dots \erefz v_m \erefz v_a $. Denoted by $v_c \eancestor v_a$.
For simplicity, we simply use $v_c \eancestor v_s$ to refer both ancestor and self-ancestor relationship, unless we need to distinguish the two cases.

{\bf S-OPERA chain:}
An S-OPERA chain is the local view of the DAG $G$=($V$,$E$). Each vertex $v_i \in V$ is an event block. Each block has a weight, which is the \emph{validation score}. An edge ($v_i$,$v_j$) $\in E$ refers to a hashing reference from $v_i$ to $v_j$; that is, $v_i \erefz v_j$.

\dfnn{S-OPERA chain}{In \stakedag, a S-OPERA chain is weighted DAG stored on each node.}

\dfnn{Leaf}{The first created event block of a node is called a leaf event block.}

\dfnn{Root}{An event block $v$ is a root if either (1) it is the leaf event block of a node, or (2) $v$ can reach more than $2W/3$ validating power from previous roots.}

\dfnn{Root set}{The set of all first event blocks (leaf events) of all nodes form the first root set $R_1$ ($|R_1|$ = $n$). The root set $R_k$ consists of all roots $r_i$ such that $r_i$ $\not \in $ $R_i$, $\forall$ $i$ = 1..($k$-1) and $r_i$ can reach more than $2W/3$ validating power from other roots in the current frame, $i$ = 1..($k$-1).}

\dfnn{Frame}{Frame $f_i$ is a natural number that separates Root sets. The root set at frame $f_i$ is denoted by $R_i$.}

\dfnn{Creator}{If a node $n_i$ creates an event block $v$, then the creator of $v$, denoted by $cr(v)$, is $n_i$.}

\dfnn{Clotho}{A root $r_k$ in the frame $f_{a+3}$ can nominate a root $r_a$ as Clotho if more than 2n/3 roots in the frame $f_{a+1}$ dominate $r_a$ and $r_k$ dominates the roots in the frame $f_{a+1}$.}

\dfnn{Atropos}{An Atropos is a Clotho that is decided as final.}

Event blocks in the subgraph rooted at the Atropos are also final events. Atropos blocks form a Main-chain, which allows time consensus ordering and responses to attacks.

\subsection{Layering}\label{se:layering}

For a directed acyclic graph $G$=($V$,$E$), a layering is to assign a layer number to each vertex in $G$.

\dfnn{Layering}{A layering (or levelling) of $G$ is a topological numbering $\phi$ of $G$, $\phi: V \rightarrow Z$,  mapping the  set  of  vertices $V$ of $G$ to  integers  such  that $\phi(v)$ $\geq$ $\phi(u)$ + 1 for every directed edge ($u$, $v$) $\in E$.}

\dfnn{Hierarchical graph}{
	For a layering $\phi$, 
	the produced graph $H$=($V$,$E$,$\phi$) is a \emph{hierarchical graph}.}

H is also called an $l$-layered directed graph and could be represented as
$H$=($V_1$,$V_2$,$\dots$,$V_l$;$E$).
Let $V^{+}(v)$ denote the outgoing neighbours of $v$; $V^{+}(v)$ =$\{u \in V | (v,u) \in E\}$. Let $V^{-}(v)$ denote the incoming neighbours of $v$; $V^{-}(v)$=$\{ u \in V | (u,v) \in E \}$.
The height of $H$ is $l$. The width of $H$ is the number of vertices in the longest layer; that is, $max_{1 \leq i \leq l} |V_i|$.

There are several approaches to DAG layering. The two most common ones are given as follows.
\emph{Longest path layering} is a list scheduling algorithm produces hierarchical graph with the smallest possible height~\cite{bang2008digraphs,Sedgewick2011}. For minimum height, it places all source vertices in the first layer $V_1$. The layer $\phi(v)$ for every remaining vertex $v$ is recursively defined by $\phi(v)$ = $max\{\phi(u) | (u, v) \in E\}$ + 1.  
By using a topological ordering of the vertices~\cite{Mehlhorn84a}, the algorithm can be implemented in linear time $O$($|V|$+$|E|$).
\emph{Coffman-Graham} (CG) algorithm, which considers a layering with a maximum width~\cite{Coffman1972}. It is currently the most commonly used layering method. Lam and Sethi~\cite{Spinrad1985} showed that the number of layers $l$ of the computed layering with width $w$ is bounded by $l \leq (2-2/w).l_{opt}$, where $l_{opt}$ is the minimum height of all layerings with  width $w$. So, CG algorithm is an exact algorithm for $w \leq 2$. In certain appplications, the notion of width does not consider dummy vertices.

{\bf Online Layering in \stakedag:}
Algorithm~\ref{algo:onlinelayering} presents the Online layering algorithms. Online Longest Path Layering (O-LPL) algorithm takes as input the following information: $U$ is the set of processed vertices, $Z$ is the set of already layered vertices, $l$ is the current height (maximum layer number), $V'$ is the set of new vertices, and $E'$ is the set of new edges.
The algorithms has a worst case time complexity of $O$($|V'|$+ $|E'|$).

\begin{algorithm}
	\footnotesize
	\caption{Online Layering Algorithms}\label{algo:onlinelayering}
	\begin{minipage}{0.45\textwidth}
	\begin{algorithmic}[1]
		\State Require:A DAG $G$=($V$,$E$) and $G'$=($V'$,$E'$)
		\Function{OnlineLPL}{$U$, $Z$, $l$, $V'$, $E'$}
		\State $V_{new}$ $\gets$ $V \cup V'$, $E_{new}$ $\gets$ $E \cup E'$
		\While{$U \neq V_{new}$}
		\State Select $v \in  V_{new} \setminus U $ with $V_{new}^{+}(v) \subseteq Z$
		\If{$v$ has been selected}	
		\State $l' \gets max\{\phi(u) | u \in V^+_{new}(v) \}$ +1	
		\State $\phi(v)$ $\gets l'$ 
		\State $U \gets U \cup \{v\}$
		\EndIf
		\If{no vertex has been selected}		
		\State $l$ $\gets$  $l$ + 1
		\State $Z \gets Z \cup U$
		\EndIf
		\EndWhile
		\EndFunction
	\end{algorithmic}
	\end{minipage}
\begin{minipage}{0.55\textwidth}
	\begin{algorithmic}[1]
		\State Require:A reduced DAG $G=(V,E)$ and $G'=(V',E')$
		\Function{OnlineCoffmanGraham}{$W$, $U$, $l$, $\phi$, $V'$, $E'$}
		\ForAll{$v \in V'$}
		$\lambda(v) \gets \infty$
		\EndFor
		\For{$i$ $\gets$ 1 to $|V'|$}
		\State Choose $v \in V'$ with $\lambda(v) = \infty$ s.t. $|V^{-}(v)|$ is min
		\State $\lambda(v) \gets 1$
		\EndFor
		\State $V_{new} \gets V \cup V'$; $E_{new} \gets E \cup E'$
		\While{$U \neq V_{new}$}
		\State Choose $v \in V_{new} \setminus U$ s.t. $V_{new}^{+}(v) \subseteq U$ and $\lambda(v)$ is max
		\State $l' \gets max\{\phi(u) | u \in V^+_{new}(v) \}$ +1			
		\If{$|\phi_{l'}| \leq W$ and $V_{new}^{+}(v) \subseteq \phi_1 \cup \phi_2 \cup \dots \cup \phi_{l'-1}$}
		\State $\phi_{l'} \gets \phi_{l'} \cup \{v\}$
		\Else
		\State	$l$ $\gets$ $l$+1
		\State  $\phi_l \gets \{v\}$
		\EndIf
		\State $U \gets U \cup \{v\}$
		\EndWhile
		\EndFunction
	\end{algorithmic}
\end{minipage}
\end{algorithm}

Online Coffman-Graham (O-CG) Algorithm takes as input the following: $W$ is the fixed maximum width, $U$ is the unprocessed vertices, $l$ is the current height (maximum layer), $\phi$ is the layering assignment, $V'$ is the set of new vertices, $E'$ is the set of new edges. The algorithms has a worst case time complexity of $O(|V'| ^2)$.
The algorithm assumes that the given value of width $W$ is sufficiently large. Section~\ref{sec:layerwidth} gives some discussions about choosing an appropriate value of $W$.

\subsection{Semantics and Proofs of Consensus}

$S_{\phi}$ protocol uses several novel concepts such as S-OPERA chain, H-OPERA chain, Root graph and Frame assignment to achieve deterministic consensus in PoS DAG.
We now present our formal model for the consistency of knowledge across the distributed network of nodes. We use the consistent cut model, based on our ONLAY paper~\cite{onlay19}. We generalize the concurrent common knowledge in PoS DAG-based protocol, for details see the original CCK paper~\cite{cck92}.

For an S-OPERA chain $G$, let $G[v]$ denote the subgraph of $G$ that contains nodes and edges reachable from $v$.

\dfnn{Consistent chains}{Two chains $G_1$ and $G_2$ are consistent, denoted by $G_1 \sim G_2$, if for any event $v$ contained in both chains, $G_1[v] = G_2[v]$.}

For any block $v$, a node must already have the $k$ references of $v$ in order to accept $v$.
For any two nodes, suppose their S-OPERA chains contain the same event $v$. Thus, both S-OPERA chains must contain $k$ references of $v$. Presumably, the cryptographic hashes are  always secure and references must be identical between nodes. By induction, all ancestors of $v$ must be the same. Hence, the two consistent chains must contain the same set of ancestors for $v$, with the same reference and self-ref edges between those ancestors. Consequently, the two OPERA chains are consistent.\\

\dfnn{Global S-OPERA chain}{
	A global consistent chain $G^C$ is a chain such that $G^C \sim G_i$ for all $G_i$.}

Let $G \sqsubseteq G'$ denote that $G$ is a subgraph of $G'$. Some properties of $G^C$ are given as follows: (1) $\forall G_i$ ($G^C \sqsubseteq G_i$); (2)
$\forall v \in G^C$ $\forall G_i$ ($G^C[v] = G_i[v]$); (3)
($\forall v_c \in G^C$) ($\forall v_p \in G_i$) (($v_p \hbefore v_c) \Rightarrow v_p \in G^C$).

The layering of consistent S-OPERA chains is consistent itself.

\dfnn{Consistent layering}{
	For any two consistent OPERA chains $G_1$ and $G_2$,  layering results $\phi^{G_1}$ and $\phi^{G_2}$ are consistent, denoted by $\phi^{G_1} \sim \phi^{G_2}$, if $\phi^{G_1}(v) = \phi^{G_2}(v)$, for any  vertex $v$ common to both chains.}

\begin{thm}	
	For two consistent OPERA chains $G_1$ and $G_2$, the resulting H-OPERA chains using layering $\phi_{LPL}$ are consistent.
\end{thm} 

The theorem states that for any event $v$ contained in both OPERA chains, $\phi_{LPL}^{G_1}(v) = \phi_{LPL}^{G_2}(v)$. Since $G_1 \sim G_2$, we have $G_1[v]= G_2[v]$. Thus, the height of $v$ is the same in both $G_1$ and $G_2$. Thus, the assigned layer using $\phi_{LPL}$ is the same for $v$ in both chains.

\begin{prop}[Consistent root graphs]
	Two root graphs $G_R$ and $G'_R$ from two consistent H-OPERA chains are consistent.
\end{prop}

\dfnn{Consistent root}{
	Two chains $G_1$ and $G_2$ are root consistent, if for every $v$ contained in both chains, $v$ is a root of $j$-th frame in $G_1$, then $v$ is a root of $j$-th frame in $G_2$.}

\begin{prop}
	For any two consistent OPERA chains $G_1$ and $G_2$, they are root consistent. 
\end{prop}

By consistent chains, if $G_1 \sim G_2$ and $v$ belongs to both chains, then $G_1[v]$ = $G_2[v]$.
We can prove the proposition by induction. For $j$ = 0, the first root set is the same in both $G_1$ and $G_2$. Hence, it holds for $j$ = 0. Suppose that the proposition holds for every $j$ from 0 to $k$. We prove that it also holds for $j$= $k$ + 1.
Suppose that $v$ is a root of frame $f_{k+1}$ in $G_1$. 
Then there exists a set $S$ reaching 2/3 of members in $G_1$ of frame $f_k$ such that $\forall u \in S$ ($u\hbefore v$). As $G_1 \sim G_2$, and $v$ in $G_2$, then $\forall u \in S$ ($u \in G_2$). Since the proposition holds for $j$=$k$, 
As $u$ is a root of frame $f_{k}$ in $G_1$, $u$ is a root of frame $f_k$ in $G_2$. Hence, the set $S$ of 2/3 members $u$ happens before $v$ in $G_2$. So $v$ belongs to $f_{k+1}$ in $G_2$.

Thus, all nodes have the same consistent root sets, which are the root sets in $G^C$. Frame numbers are consistent for all nodes.\\

\dfnn{Flag table}{A flag table stores reachability from an event block to the weighted roots.}

Given two consisten S-OPERA chains $G_1$ and $G_2$ ($G_1 \sim G_2$), we have the following

\dfnn{Consistent Flag Table}{For event block $v$ in both $G_1$ and $G_2$, and $G_1 \sim G_2$, then the flag tables of $v$ are consistent if they are the same in both chains.}

From the above, the root sets of $G_1$ and $G_2$ are consistent. If $v$ contained in $G_1$, and $v$ is a root of $j$-th frame in $G_1$, then $v$ is a root of $j$-th frame in $G_i$. Since $G_1 \sim G_2$, $G_1[v] = G_2[v]$. The reference event blocks of $v$ are the same in both chains. Thus the flag tables of $v$ of both chains are the same.\\

\dfnn{Consistent Validation Score}{For event block $v$ in both $G_1$ and $G_2$, and $G_1 \sim G_2$, the validation score of $v$ in $G_1$ is identical with that of $v$ in $G_2$.}

Since the flag tables of $v$ are the same on both chains, the roots are consistent. The weights of the roots are also consistent across the nodes. The validation score of an event block $v$ is the sum of the validating powers of the roots. Thus, the validation score of $v$ is also consistent across the nodes.

\begin{prop}[Consistent Clotho]
	A root $r_k$ in the frame $f_{a+3}$ can nominate a root $r_a$ as Clotho if more than 2n/3 roots in the frame $f_{a+1}$ dominate $r_a$ and $r_k$ dominates the roots in the frame $f_{a+1}$.
\end{prop}

%Each node nominates a root into Clotho via the flag table. If all nodes have an OPERA chain with same shape, the values in flag table will be equal to each other in OPERA chain. Thus, all nodes nominate the same root into Clotho since the OPERA chain of all nodes has same shape.

\begin{prop}[Consistent Atropos]
	An Atropos is a Clotho that is decided as final. Event blocks in the subgraph rooted at the Atropos are also final events. Atropos blocks form a Main-chain, which allows time consensus ordering and responses to attacks.
\end{prop}

\begin{prop}[Consistent Main-chain]
	For any two honest nodes, the Main-chain is the same on both nodes.
\end{prop}

\dfnn{Main-chain}{Main chain is a special sub-graph of the OPERA chain that stores  Atropos vertices.}

The Main chain stores the ordered list of finalized Atropos blocks. Since Clothos and Atropos blocks are consistent on any two honest nodes, the Main chain is the same on both nodes. The Main chain
can be used for fast track of finalized blocks, transaction search and for detecting potential conflicting transactions.

For global state, \stakedag\ protocol use a similar algorithm like in our ONLAY paper~\cite{onlay19}.
From local view H-OPERA chain, a node can estimate the subgraph of $H$ that other nodes have known at the time.
The intuition is that a node $n_i$, from its history $H$, conservatively knows that the other node $n_j$ may just know the event blocks under the top event of $n_j$.
Let $t_i$ denote the top event of a node $n_i$ in the local chain $H$. From the local state $H$, a node can determine all top events of all nodes in $H$.